\documentclass{article}

\usepackage[preprint]{neurips_2023}

\usepackage[utf8]{inputenc} 
\usepackage[T1]{fontenc}    
\usepackage{hyperref}       
\usepackage{url}            
\usepackage{booktabs}       
\usepackage{amsfonts}       
\usepackage{nicefrac}       
\usepackage{microtype}      
\usepackage{xcolor}         
\usepackage{amsthm}
\usepackage{amsmath}
\usepackage{color}
\usepackage{graphicx}
\usepackage{xcolor}         
\usepackage{tikz}
\usetikzlibrary{arrows.meta}
\usetikzlibrary{automata,arrows,calc,positioning}
\newtheorem{theorem}{Theorem}
\newtheorem{lemma}{Lemma}
\newtheorem{assumption}{Assumption}
\newtheorem{definition}{Definition}
\newtheorem{corollary}{Corollary}

\newtheorem{remark}{Remark}

\newcommand{\ind}{\perp\!\!\!\!\perp}
\newcommand{\noind}{\not\perp\!\!\!\!\perp}

\newcount\Comments  
\newcommand{\kibitz}[2]{\ifnum\Comments=1\textcolor{#1}{#2} \fi \ignorespaces}

\title{Self-Resolving Prediction Markets for \\ Unverifiable Outcomes}

\author{%
  Siddarth Srinivasan \\ 
  Harvard University\\
  Cambridge, MA, USA \\
  \texttt{ssrinivasan@seas.harvard.edu} \\
  \And
  Ezra Karger \\
  Federal Reserve Bank of Chicago \\
  Chicago, IL, USA \\
  \texttt{ezra.karger@chi.frb.org} \\
  \And
  Yiling Chen \\
  Harvard University \\
  Cambridge, MA, USA \\
  \texttt{yiling@seas.harvard.edu}
}
 
\begin{document}

\maketitle

\begin{abstract}
  Prediction markets elicit and aggregate beliefs by paying agents based on how close their predictions are to a verifiable future outcome. However, outcomes of many important questions are difficult to verify or \emph{unverifiable}, in that the ground truth may be hard or impossible to access. We present a novel incentive-compatible prediction market mechanism to elicit and efficiently aggregate information from a pool of agents \emph{without observing the outcome}, by paying agents the negative cross-entropy between their prediction and that of a carefully chosen reference agent. Our key insight is that a reference agent with access to more information can serve as a reasonable proxy for the ground truth. We use this insight to propose self-resolving prediction markets that terminate with some probability after every report and pay all but a few agents based on the final prediction. The final agent is chosen as the reference agent since they observe the full history of market forecasts, and thus have more information by design.  We show that it is a perfect Bayesian equilibrium (PBE) for all agents to report truthfully in our mechanism and to believe that all other agents report truthfully. Although primarily of interest for unverifiable outcomes, this design is also applicable for verifiable outcomes.
\end{abstract}

\section{Introduction}

Decision-makers often tackle complex problems based on beliefs that are difficult to validate because the ground truth is inaccessible or prohibitively expensive to access.
For example, the causal effects of policy interventions can be hard to verify, yet policy-makers propose and implement policies with complicated societal outcomes; economic and political outcomes in the distant future can be hard to predict, yet people must make plans about where to live, where to work, and how to allocate resources today; and determining that a piece of content is misinformation or a violation of terms of service is challenging, yet social media platforms make content moderation decisions to improve user experiences. 
What these examples all have in common is that in the face of unverifiable questions, decision-makers gather information from varied sources to arrive at a `best guess.'

The general problem of eliciting and aggregating disparate beliefs and data into a single `best guess' in the absence of ground truth verification is technically challenging. How do agents reason about revealing their own information, and how can we incentivize them to reveal this privately held information? In this paper, we make progress on this problem by proposing the first incentive-compatible self-resolving prediction market (or equivalently, sequential peer prediction mechanism) that simultaneously elicits and aggregates beliefs in the absence of ground truth, under nearly the same assumptions as standard prediction markets. Our proposed design sequentially elicits and aggregates probabilistic predictions, randomly terminates with some probability $\alpha$ after each prediction, and pays all but the last few agents based on the prediction of the final agent. 

\begin{figure}[h!] 
\centering
\begin{tikzpicture}[>=stealth', shorten >=1pt, auto,
    node distance=2.5cm, scale=0.6, 
    transform shape, 
    state/.style={circle, draw, minimum size=1.5cm}]
    
\begin{scope}[every node/.style={circle,thick,draw}]
    \node[state] (1) at (0,0)   {$0$};
    \node[state] (2) at (2,-3)  {$1$};
    \node[state] (3) at (4,0)   {$2$};
    \node[state] (4) at (6,-3)  {$3$};
    \node[state] (5) at (8,0)   {$T-k$};
    \node[state] (6) at (10,-3) {$T-k+1$};
    \node[state] (7) at (12,0)  {$T$} ;
    \node[state] (8) at (14,-3) {end} ;

\end{scope}

\begin{scope}[>={Stealth[black]},
              every node/.style={fill=white,circle},
              every edge/.style={draw=black,very thick}]
    \path [->,blue] (1) edge node {} (2);
    \path [->,blue] (2) edge node {$1-\alpha$} (3);
    \path [->,blue] (3) edge node {$1-\alpha$} (4);
    \path [->,blue] (4) edge node[pos=0.5, above=-6pt] {$...$} (5);
    \path [->,blue] (5) edge node {$1-\alpha$} (6);
    \path [->,blue] (6) edge node[pos=0.5, above=-6pt] {$...$} (7);
    \path [->,blue] (7) edge node {$\alpha$} (8);
    \draw[<-,black, dashed] (6) --++(0:1.65cm) node[right,text=olive]{$R$};
    \draw[<-,black, dashed] (7) --++(0:1.5cm) node[right,text=olive]{$R$};
    \path[<-,olive, dashed] (2.north) edge[bend left=100, looseness=2] node[text width=1cm,midway,above=0.3em ] {$S_{CEM}(q^{(T)}, q^{(1)})$} (7.north); 
    \path [<-,olive, dashed] (3.north) edge[bend left=90, looseness=1.3] node[text width=0.5cm,midway,above, xshift=-2.5em] {$S_{CEM}(q^{(T)}, q^{(2)})$} (7.north); 
    \path [<-,olive, dashed] (4.south) edge[bend right=95, looseness=1.6] node[text width=1cm,below=0.35em ] {$S_{CEM}(q^{(T)}, q^{(3)})$} (7.south); 
    \path [<-, olive, dashed] (5.north) edge[bend left=50, looseness=1.5] node[text width=1cm, xshift=-4em, yshift=-0.8em ] {$S_{CEM}(q^{(T)}, q^{(T-k)})$} (7.north); 

\end{scope}
\end{tikzpicture}
\vspace{-8mm}
\caption{\emph{Illustration of a self-resolving prediction market}. Each node represents an agent reporting a prediction to the mechanism, and the mechanism terminates with probability $\alpha$ after each report. Payouts for the first $T-k$ agents are determined using a negative cross-entropy market scoring rule with respect to the terminal agent $T$, while the last $k$ agents receive a flat payout $R$. $k$ can be chosen to be large enough so the mechanism is strictly truthful.} \label{fig:drawing}
\end{figure}
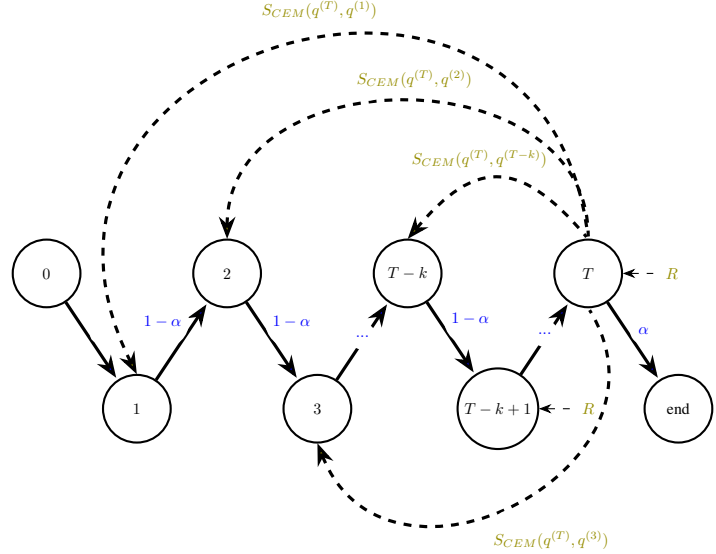

\paragraph{Contributions} Our primary contribution is a novel incentive-compatible mechanism that can elicit and aggregate truthful reports  \emph{even in the absence of ground truth}, under very similar assumptions as are applicable to standard prediction markets. Our work builds on three connected literatures: Aumann's agreement theorem, prediction markets, and peer prediction. Specifically: (1) unlike the standard framework of Aumann's agreement theorem, our mechanism \emph{provides incentives} for truthful information revelation and aggregation under the standard Aumannian protocol with many agents when agents' signals are conditionally independent given the ground truth (a kind of `informational substitutes' condition); (2) unlike prediction markets, our mechanism works even without access to the ground truth; and (3) unlike peer prediction mechanisms, our mechanism also efficiently aggregates information into a consensus prediction in the single-task setting while ensuring that it elicits minimal information, accommodates heterogeneous agents with non-binary signals, and pays zero in uninformative equilibria, as long as we have access to a sufficiently large pool of informed agents who share a common prior. Additionally, our proposal is easy to implement and intuitive to understand, with a straightforward payment scheme. In our preferred design, agents sequentially report predictions to a mechanism that terminates with some probability $\alpha$ after every prediction, and pays all but the final few agents the negative cross-entropy between their report and the final agent's report. The final few agents receive a flat fee. The primary technical challenges in paying an agent based on a reference agent's report are (1) in mitigating the incentive to strategically mislead a reference agent who observes the agent's report, and (2) in mitigating the incentive to hedge towards a common prior if the reference agent does not observe their report (as in the output agreement mechanisms by \citet{waggoner2014output}). Our key insight to mitigate these incentives is that {giving the reference agent access to enough independent informational substitutes} allows them to serve as a reasonable proxy for the ground truth. This lets us design a prediction market by making every agent's reference agent the final agent in the same prediction market as long as they are sufficiently far removed from that final agent; otherwise they just receive a flat fee. We show that truthful reporting is a strict perfect Bayesian equilibrium, and also analyze the incentive to deviate from truthful reporting when the final agent is not sufficiently distant in expectation (i.e., analyze the $\varepsilon-$PBE) under slightly milder conditions. 

\paragraph{Outline} In Section 2, we present related work; in Section 3, we present our model; in Section 4, we present some essential background; in Section 5, we investigate incentives when scoring agents based on their peers' reports; in Section 6, we propose self-resolving prediction markets; and conclude with a discussion in Section 7. All proofs are presented in Appendix \ref{app:proofs}. 

\section{Related Work}

\emph{Aumann's agreement theorem}  \citep{aumann1976agreeing} considers a simple model of a pair of agents who share a common prior and shows that common knowledge of posteriors implies equal posteriors, i.e., the agents cannot not agree to disagree. Subsequent work \citep{geanakoplos1982we, milgrom1982information, aaronson2005complexity} further formalized a standard model of rational Bayesian agents sharing a common prior who repeatedly exchange information, finding that the agents must eventually reach agreement. In particular, \citet{Frongillo2021AgreementIA} and \citet{kong2022false} find that not only must agents eventually agree, but that when agents' signals are `informational substitutes' \citep{chen2016informational} (i.e., their pieces of information have diminishing marginal value), the eventual consensus also \emph{aggregates all their private information}. While this line of work successfully introduces protocols to aggregate the private information of rational Bayesian agents into a `best guess,' it does not consider \emph{incentives}: how do we elicit \emph{true} private beliefs? Our self-resolving prediction market can be thought of as an incentive scheme for the Aumannian protocol under a conditional independence of signals assumption. 

When a mechanism designer has eventual access to the ground truth, it is possible to reward agents for reporting beliefs that correspond closely to this ground truth. Standard approaches in such a setting include proper scoring rules \citep{brier1950verification, good1952rational}, and their extension to market scoring rules \citep{hanson2003combinatorial} for use in prediction markets \citep{pennock2007computational} under certain conditions \citep{chen2010gaming, gao2013you}. \citet{ostrovsky2009information} further shows that markets can also \emph{aggregate} information for a broad class of securities. Prediction markets have shown promise in forecasting events ranging from election outcomes \citep{berg2008results} to the spread of infectious diseases \citep{polgreen2007use} to even the replicability of scientific experiments \citep{dreber2015using}, leading some to propose conditional prediction markets or \emph{decision markets} for use in decision-making \citep{hanson1999decision, othman2010decision, chen2011decision, chen2011information, chen2014eliciting}. These approaches typically run prediction markets for each possible action and use a decision rule to determine the course of action based on market predictions. However, these methods still pay out based on \emph{some} eventual outcome, require randomizing over actions (i.e., every action has a non-zero probability of being chosen), and these methods may not elicit causal information
\footnote{For example, a conditional prediction market on mortality rates with and without some health intervention might suggest that mortality is \emph{higher} under the health intervention even if the intervention works, simply because the intervention may only be politically feasible if mortality is high, i.e., causation runs the other way.} (usually what the decision-maker wishes to elicit). Unlike our mechanism, these approaches do not work in contexts where the decision-maker seeks unverifiable information. 

The \emph{information elicitation without verification} (IEWV) literature explores the problem of incentivizing agents to reveal private information when the ground truth is inaccessible, costly to access, or subjective; naturally the best one can do in such a situation is to compare agents' reports in clever ways so that truthful reporting is an equilibrium. This line of work considers wide-ranging applications such as data-labeling, peer grading, online rating systems, and the evaluation of causal effects. The simplest approach is the \emph{output agreement mechanism} \citep{von2004labeling, waggoner2014output} 
that directly pays agents based on how similar their response is to that of a reference agent, but this only elicits `common knowledge' information, not agents' true beliefs. On the other hand, early major efforts to elicit \emph{true} signals in an incentive-compatible way include peer prediction \citep{miller2005eliciting} and Bayesian truth serum \citep{prelec2004bayesian}, as well as subsequent extensions \citep{kong2019information} to more general settings: (1) \emph{minimal mechanisms} requiring only a single report from the agent \citep{witkowski2013learning, kong2020information}, (2) incentive-compatibility for \emph{small populations} \citep{witkowski2012robust}, (3) allowing for \emph{heterogeneous agents} \citep{Han2021TruthfulIE}, (4) allowing for \emph{non-binary signals} to agents \citep{radanovic2013robust}, (5) \emph{incentivizing effort} \citep{witkowski2013dwelling, radanovic2016incentives}, (6) \emph{relaxing the common prior} assumption \citep{witkowski2012peer, radanovic2015incentives}, (7) explicitly modeling \emph{adversarial agents} \citep{schoenebeck2021information}, and (8) reducing the appeal of \emph{uninformative equilibria} \citep{schoenebeck2020two, shnayder2016informed, Agarwal2020PeerPW, liu2020surrogate}, where every agent simply makes the same report regardless of their private signal. The challenge of persistent uninformative equilibria can be particularly significant \citep{gao2016incentivizing}; when a mechanism rewards an agent when their response matches a peer's, how do we prevent all agents from just making the same report regardless of their private signal to maximize reward? The most common approaches assume access to a large number of IID tasks (the \emph{multi-task setting}) and leverage statistical relationships between agent reports to ensure that uninformative equilibria pay less than informative equilibria. In our harder \emph{single-task setting}, recently proposed solutions  \citep{kong2019information, schoenebeck2020two, srinivasan2021auctions} typically rely on paying agents based on some kind of `improvement' over peers; these works are some of the most general solutions to the elicitation problem we consider, but do not handle aggregation and can be impractical for eliciting information from a large number of agents. In general, the IEWV literature focuses on \emph{eliciting} truthful information, and does not typically consider the problem of \emph{aggregating} information. \citet{frongillo2015elicitation} do consider elicitation for aggregation, but make strong parametric assumptions on agents' information structures. Despite `sequential' mechanisms where agents observe and update on other agents' reports being a natural way to aggregate information, the IEWV literature has primarily focused on eliciting information from agents in `parallel';  existing work on sequential peer prediction \citep{feng2022peer, liu2017sequential} primarily explores learning dynamics. By contrast, our proposed mechanism is an incentive-compatible sequential mechanism that efficiently aggregates information in the single-task setting and pays zero in uninformative equilibria. We also note that empirical evaluation of such mechanisms has been more limited, with mixed results \citep{gao2014trick, karger2021reciprocal, debmalya2020effectiveness} that depend on specific design details like the presence of uninformative equilibria or effortful agents. 

Finally, we note some other related strands of work. \citet{freeman2017crowdsourced} propose decentralized outcome resolution for a prediction market by separately eliciting votes for the true outcome from `arbiters' using peer prediction; our work embeds the resolution directly into the mechanism and does not require the eventual outcome to be observable at all. Our model for how agents observe and update their beliefs is similar to the sequential social learning models in the social learning literature \citep{golub2017learning}. \citet{mclean2002informational} establish a conceptually similar result to ours -- efficient allocation can be approximately incentive-compatible even when the state of the world is not observed if agents are `informationally small' (i.e., no single agent’s private information significantly alters the aggregate inference of the state). However, since agents’ utilities are directly tied to the true state in their model (as opposed to being determined solely by the mechanism's payoff structure), their mechanism does not face the challenge of uninformative equilibria where agents make uniform reports, which is a `default' behavior in pure information elicitation settings like peer prediction. Thus, our work can be thought of as engineering the idea of `informational smallness' into a sequential mechanism with market scoring rules primarily in a \emph{peer prediction setting}.

\section{Model}

Our goal is to elicit and aggregate continuous, probabilistic predictions about some binary\footnote{In the case of categorical outcomes, one can run a $1$-vs-all prediction market, and for continuous outcomes, one can discretize the outcome space and run binary prediction markets as well. {Care should be taken in interpreting such predictions; if the outcomes follow a fat-tailed distribution, tail probabilities may be hard to elicit and could have catastrophic consequences \citep{taleb2023probability}}.} \emph{Yes/No} outcome $Y$ from agents with heterogeneous beliefs. We write the space of outcomes as $\Omega_Y = \{0, 1\}$ and index outcome $\omega_i \in \Omega_Y$ as $Y_i$. Every agent $j$ has a discrete and finite signal space $\Omega_{\hat{X}^{(j)}}$, and receives one of $N_j$ possible signals, with $N_j \in \mathbb{Z}_+$.

\paragraph{Peer Prediction} We consider an agent $t$ and their reference agent $r$ from whom we will elicit predictions about some outcome $Y$. Agent $t$ has access to their own private signal ${X}^{(t)}$ not available to agent $r$. Similarly agent $r$ has access to private signal $X^{(r)}$ not available to agent $t$. Both agents have access to some shared signal $X^{(s)}$ is common knowledge. Thus, the full set of information available to agents $t$ and $r$ is $\hat{X}^{(t)} = \{ X^{(t)}, X^{(s)} \}$ and $\hat{X}^{(r)} = \{X^{(r)},  X^{(s)}\}$ respectively. We assume that agents' signals $x^{(j)} \in \Omega_{X^{(j)}}$ are drawn from a common prior $\mathbb{P}(Y, X^{(r)}, X^{(t)}, X^{(s)})$. Upon observing their private signals, agent $t$ arrives at their posterior over $Y$ through Bayesian updating as ${\bf p}^{(t)} = \mathbb{P}(Y|x^{(t)}, x^{(s)})= \mathbb{P}(Y|\hat{x}^{(t)})$. Agent $t$ may then strategically report their posterior to the mechanism as ${\bf q}^{(t)} = \sigma\left(\hat{x}^{(t)} \right)$ for some strategy $\sigma: {\Omega_{\hat{X}^{(t)}} } \rightarrow \Delta_{\Omega_Y}$. 
Agent $r$ may or may not observe agent $t$'s report to the mechanism, and they arrive at a posterior $\tilde{\bf p}^{(r)} = \mathbb{P}(Y|\hat{x}^{(r)}, {\bf q}^{(t)})$ or $\tilde{\bf p}^{(r)} = \mathbb{P}(Y|\hat{x}^{(r)})$ accordingly. Agent $r$ then reports their own posterior to the mechanism ${\bf q}^{(r)}$, and the mechanism pays agent $t$ the negative cross-entropy between the reference agent's prediction ${\bf q}^{(r)}$ and  their prediction ${\bf q}^{(t)}$.

\paragraph{Prediction Market} We consider a prediction market that elicits predictions from $M$ agents, arriving one at a time. Each agent participates exactly once at time-step $t$, where they report a prediction ${\bf q}^{(t)}$ about the outcome $Y$ to the mechanism. The agent at time-step $t$ has access to a private signal $X^{(t)} = x^{(t)}$ drawn from a joint prior distribution over outcome and signals $\mathbb{P}(Y, X^{(1)}, \ldots, X^{(M)})$, as well as all previous reports to the mechanism ${\bf q}^{(1:t-1)} = \{{\bf q}^{(1)}, \ldots {\bf q}^{(t)-1}\}$. Agent $t$ conditions on this information $\hat{X}^{(t)} = \{X^{(t)}, {\bf q}^{(1:t-1)}\}$ to arrive at a posterior $\tilde{\bf p}^{(t)} = \mathbb{P}(Y|x^{(t)}, {\bf q}^{(1:t-1)})$ over $Y$. Agent $t$ may then strategically report their posterior to the mechanism as ${\bf q}^{(t)} = \sigma\left(\hat{ x}^{(t)}\right)$ for some strategy $\sigma: {\Omega_{\hat{X}^{(t)}}} \rightarrow \Delta_{\Omega_Y}$. After each agent $t$ makes their report, the market terminates with probability $\alpha$; for all but the last $k$ agents (who receive a flat payout), the terminal agent in the same prediction market is the reference agent $r$. The reference agent $r$ observes all previous forecasts in the prediction market and arrives at a posterior $\tilde{\bf p}^{(r)} = \mathbb{P}(Y|{x}^{(r)}, {\bf q}^{(1:r-1)})$. The mechanism then pays agent $t$ based on how much their prediction ${\bf q}^{(t)}$ \emph{improves} the negative cross-entropy with respect to the reference agent's prediction ${\bf q}^{(r)}$ compared over the previous agent's prediction ${\bf q}^{(t-1)}$.

\paragraph{Notation} We use index $t$ to refer to the primary agent we are considering, index $r$ to refer to the reference agent whose prediction will be used to determine agent $t$'s payoff, and index $j$ to refer to an arbitrary agent. We denote an agent $j$'s truthful but potentially misinformed posterior (from observing other agents' reports and believing those to have been truthful) as $\tilde{\bf p}^{(j)}$. If the reports of other agents were in fact truthful, we denote agent $j$'s truthful and accurately informed posterior as ${\bf p}^{(j)}$. Similarly, we denote an agent $j$'s strategic and potentially misinformed report as  $\tilde{\bf q}^{(j)}$ and their strategic but accurately informed report as ${\bf q}^{(j)}$. We also use ${\bf r}$ to denote a reference agent $r$'s reported posterior when the details of how they arrived at their posterior do not matter. Generally, we use superscripts to indicate agents and subscripts to index a particular signal/outcome. We denote the common prior probability that $Y=1$ and $Y=0$ as $y_1$ and $y_0$ respectively, and the prior vector as ${\bf y}$. 

\subsection{Assumptions} Here we present some standard assumptions about agents' signals and the information they possess. The first two assumptions have to do with the common knowledge that potentially heterogeneous agents have regarding their priors and Bayesian rationality. These assumptions let us treat agents' reactions to new information with Bayes' rule in a consistent manner.

\begin{assumption}[Common Knowledge of Rationality and Risk Neutrality]
We assume that the agents are risk neutral and fully rational Bayesian agents and this is common knowledge.
\end{assumption}

\begin{assumption}[Common Knowledge of Common Prior]\label{as:cp}
We assume that the agents' signals $X^{(t)}$ and outcome $Y$ are drawn from a common prior $\mathbb{P}(Y, X^{(1)}, \ldots, X^{(M)})$ that is common knowledge.
\end{assumption}

Our next assumption restricts the class of information structures we study. The common \emph{stochastic relevance} assumption amounts to claiming that every agent's signal induces a unique posterior over $Y$. This allows agents to invert other agents' posteriors to infer their signal. Without this property, it would be challenging to reason about what agents know when they see another agent's forecast. This is a standard assumption, and it is plausible; practically, if two different signals induce the same posterior, we can simply ‘group’ and relabel the two signals as a single signal, and update the prior accordingly. This is also sometimes referred to as  `distinct signals'  assumption.
 
\begin{assumption}[Stochastic Relevance]\label{as:stoch}
    We assume that all agents' signals are stochastic relevant with respect to the event of interest, i.e.,  $\mathbb{P}(Y|X^{(j)}=x) \neq \mathbb{P}(Y|X^{(j)} = \tilde{x})$ for all $x \neq \tilde{x}$.
\end{assumption}

Next, we restrict our attention to information structures where agents' signals are conditionally independent given the ground truth. This essentially amounts to an assumption that agents' signals are \emph{informational substitutes} \citep{chen2016informational}, i.e., signals' informativeness has diminishing marginal value. Although this restricts the class of information structures we study, these are the same class of information structures for which we can guarantee that \emph{standard} prediction markets (where agents paid based on verifiable outcomes) are incentive-compatible \citep{chen2010gaming}. 

\begin{assumption}[Conditional Independence] \label{as:is}
    We assume that agents' signals are independent given the outcome $Y$, i.e., $X^{(j)} \ind X^{(j')} \mid Y$ for all $j \neq j'$. 
\end{assumption}

Before stating our final assumption, we introduce a notion of `distance' between distributions we will use to compare $\mathbb{P}(X^{(t)}|Y=0)$ and $\mathbb{P}(X^{(t)}|Y=1)$: the Bhattacharyya coefficient. 

\begin{definition}[$(1-\delta^{(j)})$-Bhattacharyya Coefficient]\label{def:bhat}
    The Bhattacharyya coefficient of agent $j$'s information structure with conditional probability table $\mathbb{P}(X^{(j)}|Y)$ is  $$\text{BC}\left(\mathbb{P}(X^{(j)}|Y=0), \mathbb{P}(X^{(j)}|Y=1) \right) = \sum_{i=1}^{N_j} \sqrt{\mathbb{P}(X_i^{(j)}|Y=0)\cdot \mathbb{P}(X_i^{(j)}|Y=1)} = {1- \delta^{(j)}}$$ Equivalently, the Bhattacharyya distance between the conditional distributions is  $$D_B\left( \mathbb{P}(X^{(j)}|Y=0), \mathbb{P}(X^{(j)}|Y=1) \right) = -\ln\left({1-\delta^{(j)}}\right)$$
\end{definition}

Observe that when the Bhattacharyya coefficient $(1-\delta^{(j)})$ is strictly less than 1, it means the columns of $\mathbb{P}(X^{(j)}|Y)$ are not identical, i.e., $\mathbb{P}(X^{(j)}|Y=0) \neq \mathbb{P}(X^{(j)}|Y=1)$ or equivalently $X^{(j)} \noind Y$. Indeed, we have the property that $1-\delta^{(j)} < 1$ for all $j$ 
by stochastic relevance (Assumption \ref{as:stoch}); at most one row of $\mathbb{P}(X^{(j)}|Y)$ can have identical probabilities, otherwise multiple signals would produce a likelihood ratio of 1 and hence identical posteriors. 

With this context, we introduce our final assumption of $(\delta, \eta)-$informative signals -- our only departure from the standard set of assumptions made in the prediction markets literature, though such assumptions have been studied in the social learning literature \citep{golub2017learning}. $\delta-$informativeness is essentially a `global minimum informedness' condition, requiring that 
all agents' Bhattacharyya coefficients are upper bounded by some $1-\delta$, i.e., $\sup_j 1- \delta^{(j)} = 1 - \delta < 1$. $\eta-$informativeness is a strengthening of the common assumption that all agents' signals have non-zero probability, requiring a `global minimum' probability of any signal, i.e., $\inf_t \min_{i, \omega} \mathbb{P}(x^{(j)}_i|Y=\omega) = \eta > 0$. This assumption ensures that the information gained by adding agents to our prediction market doesn't diminish too quickly. This is a plausible assumption since in practice, it is (1) highly unlikely that agents will have arbitrarily powerful/informative signals -- $\eta$-informativeness asserts a sensible constraint that there is an upper limit on how much a signal can update an agent's prior; and (2) highly unlikely that agents will have arbitrarily uninformative signals -- $\delta$-informativeness asserts a lower limit on how close agents' sampling distributions can be for $Y=1$ vs. $Y=0$.

\begin{assumption}[All Agents Have $(\delta, \eta)$-informative Signals]
We assume that all agents' signals are $(\delta, \eta)$-informative, i.e., we have $\text{BC}\left(\mathbb{P}(X^{(j)}|Y=0), \mathbb{P}(X^{(j)}|Y=1) \right) \leq {1-\delta}$ for all $j$ for some $\delta > 0$, and also that  $\min_{x, \omega} \mathbb{P}(X^{(j)}=x | Y=\omega) \geq  \eta$ for all $j$ and some $\eta > 0$.\label{as:deleta} 
\end{assumption}

Finally, we introduce the notion of $\tau-$granularity, which captures the smallest distance between the log likelihood ratios of two distinct signals for a given agent. This is always defined since the signal space is discrete and finite. Small $\tau$ implies that the agent has at least two signals that induce similar posteriors, while larger $\tau$ implies that the agent's posteriors are substantially different under different signals. This determines the degree to which an agent is inclined to pass off one signal as the other, so we will use it to establish our mechanism's strict truthfulness. For convenience, we treat $\tau$ as the same for all agents; our analysis of strict truthfulness still applies if this differs across agents, but the condition for strict truthfulness would vary by agent's $\tau$-granularity.
\begin{definition}[$\tau-$granularity]
    Let $\lambda_x = \log\left(\frac{\mathbb{P}(x|Y=1)}{\mathbb{P}(x|Y=0)}\right)$. The granularity of an agent $j$'s signal space is $\tau^{(j)}$ where $\tau^{(j)} = \min_{x, x' \in \Omega_{X^{(j)}}}|\lambda_x - \lambda_{x'}|$.
\end{definition}

\section{Background}

We now present background on scoring rules and prediction markets to set up our main results.

\begin{definition}[Scoring Rules, Proper Scoring Rules, and Strictly Proper Scoring Rules] \label{def:sr}
A scoring rule is a function $S: \Omega \times \Delta_\Omega \rightarrow \mathbb{R}$ that scores a probabilistic prediction ${\bf q} \in \Delta_\Omega$ against an outcome $\omega_i \in \Omega$. A scoring rule is \textit{proper} if, whenever the outcome $\omega_i$ is drawn from the distribution ${\bf p} \in \Delta_{\Omega}$, the expected score $\mathbb{E}_{i \leftarrow {\bf p}}[S(i, {\bf p})] \geq \mathbb{E}_{i \leftarrow {\bf p}}[S(i, {\bf q})]$ for any ${\bf q} \neq {\bf p}, {\bf q} \in \Delta_\Omega$. If this inequality is strict, the scoring rule is \textit{strictly proper}.
\end{definition}

Strictly proper scoring rules incentivize truthful reports of the belief of the distribution that outcomes $\omega$ are drawn from, as this strictly maximizes the expected reward. A commonly used scoring rule is the log scoring rule, which is the basis for the cross-entropy scoring rule used throughout this paper.

\begin{definition}[Log Scoring Rule]
The log scoring rule is a proper scoring rule where $S(i, {\bf q}) = \log q_i$.
\end{definition}

 Next, we present the \emph{cross-entropy scoring rule}, which we will use to score agents against each other (like in peer prediction) with the goal of resolving agent payoffs without ground truth. 

\begin{definition}[Cross-Entropy Scoring Rule]\label{def:xent}
    Given an agent's report ${\bf q}$ and a reference prediction ${\bf r} \in \Delta_\Omega$, the cross-entropy scoring rule is the negative cross-entropy between the distributions, i.e., $S_{CE}({\bf r}, {\bf q}) = - H({\bf r}, {\bf q}) = \sum_i r_i \log q_i$.
\end{definition}

 The cross-entropy scoring rule is not quite a scoring rule in the sense of Definition \ref{def:sr}, since it scores an agent's probabilistic prediction against another probabilistic prediction, as opposed to a discrete outcome.
 While we could sample from the reference agent's prediction and use the log scoring rule to pay the agent, this would unnecessarily increase the variance in payouts. Instead, we can just `take the expectation' with respect to the log scoring rule, which amounts to a negative cross-entropy payoff: $\mathbb{E}_{i \leftarrow {\bf r}}[\log {q}_i] = \sum_i r_i \log q_i = - H({\bf r}, {\bf q})$. 
As such, it inherits the property that for any fixed $\textbf{r}$, we have $\max_{\bf q} S_{CE}({\bf r}, {\bf q}) = S_{CE}({\bf r}, {\bf r})$. Of course, agents do not know what the reference agent will report, and so seek to maximize their \emph{expected payoff} by reporting their expectation of ${\bf r}$: $\mathbb{E}_{{\bf r}}[S_{CE}({\bf r}, {\bf q}^{(t)})] = \mathbb{E}\left[\sum_i r_i \log \left(q^{(t)}_i\right) \right] = \sum_i \mathbb{E} \left[ r_i   \right]\log \left({q^{(t)}_i}\right) = -H\left(\mathbb{E} \left[ {\bf r}  \right], {\bf q}^{(t)} \right)$. 

\paragraph{Prediction Markets} In a standard prediction market, agents may buy and sell securities that pay out a fixed amount contingent on the occurrence of some future event, usually facilitated through a continuous double auction (CDA) or automated market maker (AMM). \citet{chen2012utility} showed that a prediction market with a cost function-based automated market maker where agents change market prices by making trades can equivalently be implemented with \emph{market scoring rules} (MSR) \citep{hanson2003combinatorial} where agents report their forecast directly to the mechanism and are paid based on how much they \emph{improve} the market prediction. This work will use MSRs with cross-entropy scoring rules to directly elicit predictions from agents and determine payoffs, although we show how our mechanism can be cast as one in which agents trade securities in Section \ref{app:cfmm}.

\begin{definition}[$S$-Market Scoring Rule]
An $S-$market scoring rule is a sequential and shared proper scoring rule that pays agents the difference in score (according to the given proper scoring rule $S$) between their forecast and the prior forecast, i.e., the payoff is $ S(i, {\bf q}^{(t)}) - S(i, {\bf q}^{(t-1)}) $ where $S$ is the given proper scoring rule. 
\end{definition}

\begin{definition}[Cross-Entropy Market Scoring Rule (CE-MSR)]
    Given agent $t$'s prediction ${\bf q}^{(t)}$, market prediction ${\bf q}^{(t-1)}$ and a reference prediction ${\bf r}$, the cross-entropy market scoring rule is $S_{CEM}({\bf r}, {\bf q}^{(t)}, {\bf q}^{(t-1)}) = - H({\bf r}, {\bf q}^{(t)}) + H({\bf r}, {\bf q}^{(t-1)}) = \sum_i r_i \log \left(\frac{q^{(t)}_i}{q^{(t-1)}_i}\right) $.
\end{definition}\label{eq:ptsr}

Just like with cross-entropy scoring rules, the CE-MSR incentivizes agents to predict their reference agent's \emph{expected prediction}, since $\mathbb{E}_{{\bf r}}[S_{CEM}({\bf r}, {\bf q}^{(t)}, {\bf q}^{(t-1)}) ] = -H\left(\mathbb{E} \left[ {\bf r} \right], {\bf q}^{(t)} \right) +H\left(\mathbb{E} \left[ {\bf r} \right], {\bf q}^{(t-1)} \right)$, which is maximized by setting ${\bf q}^{(t)} = \mathbb{E}[{\bf r}]$. 

\paragraph{Informational Substitutes} Building on Definition \ref{as:is}, we further present two key intuitions about our conditional independence assumption (equivalent to a kind of informational substitutes assumption). First, observe that under such an assumption, in the standard Aumannian protocol where $n$ agents announce their posteriors in sequence (Observation B.1 by \citet{kong2022false}), every agent's posterior is equal to $\mathbb{P}(Y|X_1, \ldots, X_n)$ after a single round; agents do not need to repeatedly exchange posteriors. In other words, whenever an agent announces their beliefs, previous agents learn new information but have nothing new to contribute. This is the same reason why we have an `all-rush' equilibrium in standard prediction markets \citep{chen2010gaming} where every agent rushes to reveal their private information and only needs to participate once in the mechanism. This allows our mechanism to aggregate historical information with a single round of elicitation. Second, observe that when combined with stochastic relevance (Assumption \ref{as:stoch}), an agent can invert a sequence of predictions to obtain the sequence of signals that each agent observed. This allows agents to be confident that subsequent agents who observe their prediction will know everything they know, and so be able to aggregate \emph{all} historical information efficiently (see Lemma 4 by \citet{chen2010gaming}). 

\paragraph{Perfect Bayesian Equilibrium} A Perfect Bayesian Equilibrium (PBE) is game-theoretic solution concept for sequential Bayesian games with incomplete information that specifies a strategy profile (every agent's action for every possible information set) and belief profile (every agent's belief at every information set arrived at through Bayes' rule where possible) such that no agent can gain in expectation by unilateral deviation. A PBE must also specify off-the-equilibrium-path beliefs consistent with the strategy profile, i.e., beliefs formed after an unexpected move, to ensure that the equilibrium captures not only optimal behavior but also reasonable expectations in every possible contingency. An $\varepsilon-$PBE specifies a strategy and belief profile under which no agent can gain more than $\varepsilon$ payoff in expectation by unilateral deviation.

\section{Incentives under Cross-Entropy Scoring Rules}  \label{sec:peerpred}

In this section, we investigate agents' beliefs about other agents' beliefs and agents' incentives under cross-entropy scoring rules, independent of any particular prediction market design. Specifically, we study the natural question of whether there is a mismatch between what cross-entropy scoring rules incentivize (predicting the reference agent's prediction) and what the mechanism designer wants to elicit (beliefs about some random variable $Y$). First, we establish the number of agents required so that truthful reporting strictly maximizes the agent's expected payoff. Subsequently, we establish an upper bound on agents' gain in reward by deviating from truthful reporting under a slightly milder condition: when $\tau-$granularity is not known.

\paragraph{Setup} Suppose agents $t$ and $r$ have private information $x^{(t)}$ and $x^{(r)}$ respectively, and shared information $x^{(s)}$ that is common knowledge. The agents share a common prior $\mathbb{P}(Y, X^{(r)}, X^{(t)}, X^{(s)})$ over $Y$ with marginal prior over $Y$ denoted ${\bf y} = P(Y)$. 
Now, suppose agent $t$ reports their posterior over $Y$ as ${\bf q}^{(t)}$. If agent $r$ observes ${\bf q}^{(t)}$ associated with some signal $X^{(t)} = \tilde{x}^{(t)}$, agent $r$ updates their beliefs as if they observed $X^{(t)} = \tilde{x}^{(t)}$ (agents can invert posteriors to obtain the signal under Assumption \ref{as:stoch}). If ${\bf q}^{(t)}$ is not consistent with agent $t$'s information structure, agent $r$ may either ignore it or update as if they observed some arbitrary $X^{(t)} = \tilde{x}^{(t)}$. If agent $r$ does not observe this report, they take $X^{(t)} = \emptyset$ and make no update to their belief. We further assume that agent $r$ reports their true (but potentially misinformed) posterior $\tilde{\bf p}^{(r)}$.

\subsection{What Agents Believe Other Agents Believe}

Before discussing incentives, we seek to understand agents' beliefs about other agents' beliefs given the information they have. Reference agent $r$ may observe agent $t$'s report, so agent $t$ may be able to strategically mislead agent $r$ toward a preferred prediction if agent $r$ believes and updates on this information. Below, we characterize how agent $t$'s expectation of agent $r$'s belief depends on their own report and private signal.

\begin{lemma}\label{lem:post}  Agent $t$'s expectation of agent $r$'s true posterior ${\tilde{p}}_1^{(r)} = \mathbb{P}(Y=1|x^{(r)}, \tilde{x}^{(t)}, x^{(s)})$ is:
    \begin{equation}\small
       \mathbb{E}[{\tilde{p}}_1^{(r)} | x^{(t)}, x^{(s)}]  =  \mathbb{P}(Y=1|x^{(t)}, x^{(s)})  + \Delta(\Omega_{X_r}, \tilde{x}^{(t)}, {x}^{(t)}, x^{(s)})
    \end{equation}

    where $\Delta(\Omega_{X_r}, \tilde{x}^{(t)}, {x}^{(t)}, x^{(s)}) = \mu(\Omega_{X_r}, \tilde{x}^{(t)}, x^{(s)}) \cdot \rho(\tilde{x}^{(t)}, x^{(t)}, x^{(s)})$, with 
    \begin{equation}\small
        \begin{split}
            \mu(\Omega_{X_r}, \tilde{x}^{(t)}, x^{(s)}) &= \sum_{x^{(r)} \in \Omega_{X_r}}  \frac{1 }{ \frac{1}{\mathbb{P}(x^{(r)}, Y=1|\tilde{x}^{(t)}, x^{(s)})} + \frac{1}{\mathbb{P}(x^{(r)}, Y=0|\tilde{x}^{(t)}, x^{(s)})}} \\
             \rho(\tilde{x}^{(t)}, x^{(t)}, x^{(s)}) &= \frac{\mathbb{P}(Y=1|\tilde{x}^{(t)}, x^{(s)}) - \mathbb{P}(Y=1|x^{(t)}, x^{(s)})}{\mathbb{P}(Y=0|\tilde{x}^{(t)}, x^{(s)})\cdot  \mathbb{P}(Y=1|\tilde{x}^{(t)}, x^{(s)})} 
        \end{split}
    \end{equation}
\end{lemma}

\paragraph{Understanding Lemma \ref{lem:post}} Lemma \ref{lem:post} shows that agent $t$'s expectation of agent $r$'s true posterior is their own true posterior adjusted by a term $\Delta = \mu \cdot \rho$ that moves the agent's true posterior towards their posterior under $\tilde{x}^{(t)}$. Thus, agent $t$'s report influences their own expectation of agent $r$'s prediction, giving agent $t$ an additional degree of freedom. Since we wish to elicit agents' true beliefs, the next question that arises is: \emph{under what conditions does the adjustment term $\Delta$ disappear}, such agent $t$'s expectation of agent $r$'s prediction is their own true posterior? Observe that the adjustment term does not disappear by hiding agent $t$'s report from agent $r$ (i.e., $\tilde{x}^{(t)} = \emptyset$), nor by making all information private (i.e.,  $x^{(s)} = \emptyset$); this simply prompts an adjustment towards the posterior an independent observer would have without access to $\tilde{x}^{(t)}$ or $x^{(s)}$ or both. However, when an agent \emph{does} report their true posterior ($\tilde{x}^{(t)} = x^{(t)}$) and the reference agent observes and updates on this, $\rho(x^{(t)}, {x}^{(t)}, x^{(s)}) = 0$ and we have the following result:

\begin{corollary}[Martingale Property of Truthful Report]
    If agent $t$ truthfully reports their private signal $\tilde{x}^{(t)} = x^{(t)}$ and agent $r$ observes and believes this, agent $t$'s expectation of agent $r$'s posterior of $Y$ equals their own posterior of $Y$, i.e., $\mathbb{E}[{{p}}_1^{(r)} | x^{(t)}, x^{(s)}]  =  \mathbb{P}(Y=1|x^{(t)}, x^{(s)})$. \label{cor:mart}
\end{corollary}

This martingale property of posteriors is a known result, particularly in the literature exploring Aumann's agreement theorem \citep{aaronson2005complexity}. 
Yet, Corollary \ref{cor:mart} does not shed any light on how to design a mechanism in which $\Delta$ vanishes and agents are best off reporting their true beliefs. In particular, when agent $r$ updates on agent $t$'s report, agent $t$ wants to report predictions that will be close to the reference agent's prediction (achieved by truthful reporting) but \emph{also} nudge the reference agent towards greater certainty (achieved by reporting signals that induce more extreme beliefs in the reference agent). This arises from the strict convexity of negative cross-entropy, since highly certain predictions that agree have lower cross-entropy (and so higher payoff) than uncertain predictions that agree (see example in Appendix \ref{sec:exinct}). Separately, if agent $r$ does not observe agent $t$'s report, agent $t$ is incentivized to move their prediction towards the common prior, a dynamic that can lead to `coarsening' of elicited information \citep{waggoner2014output}. 
We make further progress by identifying a condition under which the adjustment term $\Delta$ diminishes: when the reference agent has access to enough informational substitutes that agent $t$ cannot access.

\begin{theorem}
    Suppose a reference agent $r$ can observe $k$ private signals $x^{(1)}, \ldots,  x^{(k)}$ that are informational substitutes, where $\Omega_{X^{(r)}} = \Omega_{X^{(1)}} \times \ldots \times \Omega_{X^{(k)}}$ and $x^{(j)} \in \Omega_{X^{(j)}}$, and agent $t$ cannot observe these signals.  Then, agent $t$'s adjustment $\Delta$ of their posterior to predict their reference agent's  true posterior is upper bounded in magnitude as $|\Delta(\Omega_{X_r}, \tilde{x}^{(t)}, x^{(t)}, x^{(s)}) |\leq \frac14 \left({{\frac{1-\eta}{\eta}}} - {{\frac{\eta}{1-\eta}}} \right) \left( 1-\delta \right)^k $.
    Consequently, we have an upper bound $|\Delta | \leq \varepsilon'$ for any $k$ where
    \begin{equation}\label{eq:kchoice} \small
    k  \geq \frac1{-\log \left( 1-\delta \right)}\log \left(\frac{1}{4\varepsilon'}\left({{\frac{1-\eta}{\eta}}} - {{\frac{\eta}{1-\eta}}} \right)\right)
    \end{equation}
       
\label{eq:imptheorem}
\end{theorem}

\paragraph{Understanding Theorem \ref{eq:imptheorem}} This result shows that we can focus on designing mechanisms that provide agent $r$ with a minimum number of $k$ informational substitutes that agent $t$ does not have access to. Observe that for any desired upper bound $\varepsilon'$ on the adjustment term $\Delta$, $k$ scales  logarithmically with $\frac1{\varepsilon'}$, logarithmically\footnote{This scaling is additionally promising because $\eta$ can get quite small as signal spaces get larger; our result shows that an exponential increase in signal space only linearly increases the need for informational substitutes.} with $\frac1{\eta}$, and inversely with the upper bound of the Bhattacharyya distance of agents' information structures, i.e., as $\frac1{-\log(1-\delta)}$. 
The key insight is that with enough informational substitutes, agent $r$'s true belief becomes a reasonable proxy for the ground truth. Thus, agent $t$ can only have a small effect on agent $r$'s belief even if agent $r$ observes their report (i.e., the $k$ informational substitutes essentially `wash out' the impact of agent $t$'s report on agent $r$).

\subsection{Peer Prediction under Cross-Entropy Scoring Rules}\label{sec:ppcesr}

Having characterized what agent $t$ expects their reference agent's belief to be,  we turn our attention to \emph{incentives} under cross-entropy scoring rules. 

Specifically, we extend our analysis to show that agent $t$ strictly maximizes their payoff by reporting their true belief when agent $r$'s has enough information. The reasoning is this: by Theorem \ref{eq:imptheorem}, agent $t$'s influence on agent $r$'s belief diminishes with the amount of information agent $r$ has, which we show translates to a diminishing influence on their payoff; given a discrete and finite signal space, agent $r$'s update will be consistent with \emph{some} signal that agent $t$ could have reported (otherwise, agent $r$ ignores it). Then, as we increase the number of informational substitutes that the reference agent has, the set of reports that agent $t$ could profitably deviate to dwindles. There finally comes a point (determined by agent $t$'s $\tau$-granularity where no signal that agent $t$ could credibly claim to observe would successfully persist to `wash out' the reference agent's information. We state this formally:

\begin{theorem}[Minimum Reference Signals for Strict Truthfulness]\label{thm:exact}
    Suppose agent $t$ is paid based on a cross-entropy scoring rule with respect to agent $r$ who reports truthfully.  Then, agent $t$ strictly maximizes their expected payoff by reporting truthfully if the reference agent has at least $k$ signals where
    \begin{equation}\small
        k  > \frac{1}{-\log \left( 1-\delta \right)}\log\left(\frac{\left|\log \left(\frac{y_1(1-\eta)^2}{y_0\eta^2}\right)\right|\left({{\frac{1-\eta}{\eta}}} - {{\frac{\eta}{1-\eta}}} \right)}{8 \left(\eta(1-\eta)\tau\right)^2} \right)
    \end{equation}
\end{theorem}

Now, this result depends on the $\tau-$granularity of agent $t$'s signal; this isn't an assumption but comes directly from our model and may not be easily known by the mechanism designer. In order to provide a result that \emph{doesn't} depend on the $\tau$-granularity, we can translate the guarantee that agent $t$'s expectation of agent $r$'s belief is within some desired $\varepsilon'$ of their own true belief to an upper bound on the gain in deviating from truthful reporting with the following: 

\begin{theorem}[Upper Bound in Deviating from Truthful Reporting under Cross-Entropy Scoring Rules]
    If agent $t$ is paid based on a cross-entropy scoring rule with respect to agent $r$ who reports truthfully, the difference in expected payoff when deviating from truthful report ${\bf p}^{(t)}$ associated with signal $x^{(t)}$ to misreport ${\bf q}^{(t)}$ is:
    \begin{equation}\small \label{eq:uboundce}
       \mathbb{E}[S_{CE}({\bf q}^{(t)}) | x^{(t)}, x^{(s)}] - \mathbb{E}[S_{CE}({\bf p}^{(t)}) | x^{(t)}, x^{(s)}] \leq \mathcal{D}_{\eta}(\Delta, {\bf y})
    \end{equation}
    where ${\bf y}$ is the prior and
    \begin{equation} \small \label{eq:deta}
        \mathcal{D}_{\eta}(\Delta, {\bf y}) = \Delta \cdot \log \left(\frac{ ({1-\eta})  y_1 +  \Delta \cdot \left(\eta y_0 + ({1-\eta}) {y_1}\right)}{\eta y_0-  \Delta\cdot \left(\eta y_0 + ({1-\eta}){y_1}\right)} \right)  
    \end{equation}
    \label{thm:bound}
\end{theorem}

\begin{remark}
As $k$ increases and $ \Delta$ vanishes (by Theorem \ref{eq:imptheorem}), the difference in expected rewards from misreporting also vanishes, i.e., $\lim_{ \Delta \to 0} \mathcal{D}_\eta(\Delta, {\bf y}) = 0$. Additionally,  $\lim_{ \Delta \to 0} \mathbb{E}[S_{CE}({\bf q}^{(t)}) | x^{(t)}, x^{(s)}] = -H({\bf p}^{(t)}, {\bf q}^{(t)})$, which is strictly maximized by truthful reporting ${\bf q}^{(t)} = {\bf p}^{(t)}$, i.e., agent $t$'s expected reward approaches the cross-entropy between their true posterior and their reported prediction.
\end{remark}

By Theorem \ref{eq:imptheorem}, we know that we can make the adjustment term vanish by giving the reference agent access to enough informational substitutes.  Now, we would ideally like to upper bound Equation \ref{eq:deta} by some desired $\varepsilon$, solve for an upper bound on $\Delta$ such that $\mathcal{D}_{\eta}(\Delta, {\bf y}) \leq \varepsilon$, and use this to determine the number of informational substitutes that agent $r$ needs through Equation \ref{eq:kchoice} in Theorem \ref{eq:imptheorem}. Unfortunately, $\mathcal{D}_\eta(\Delta, {\bf y}) = \varepsilon$ is a transcendental equation, so we can only solve for $k$ numerically. 
Thus, we numerically solve for ${\varepsilon'} = \min \{|\Delta|: \mathcal{D}_{\eta}(\Delta, {\bf y}) = \varepsilon\}$ and compute the minimum number of informational substitutes $k_{min}$ that agent $r$ requires according to Theorem \ref{eq:imptheorem} for any setting of $\delta, \eta$ and choice of $\varepsilon$. This guarantees that agent $t$ can gain no more than $\varepsilon$ by deviating from truthful reporting. We visualize the dependence of $k$ on these parameters in Appendix \ref{app:viz}, and find that parameter $\delta$ has the greatest impact on the $k$.

As a more general observation, we see from Lemma \ref{lem:post} that a simple `peer prediction' payment of paying an agent based on agreement with a reference agent is not generally truthful, because of the incentive to mislead the reference agent or hedge towards the prior. Specifically, the core incentive issue arises from the strict convexity of negative cross-entropy, since highly certain predictions that somewhat agree have lower cross-entropy (and so higher payoff) than uncertain predictions that exactly agree (see example in Appendix \ref{sec:exinct}). Yet, by Theorem \ref{thm:exact}, it becomes truthful to pay based on agreement when the reference agent has access to enough conditionally independent signals that can `wash out' the impact of agent $t$'s report on agent $r$'s belief. Although this result could be used in `parallel' peer prediction settings where agents do have $k$ private informational substitutes, we find that the most natural design for a mechanism that both elicits and aggregates information is a `sequential' peer prediction mechanism, structured as a prediction market.

\subsection{Prediction Markets with Cross-Entropy Market Scoring Rules} \label{sec:icesmr}

Given the results above, we are left with the question of how best to pick a reference agent and how to give them access to the required informational substitutes. Prediction markets are a natural solution here, since we can either pick some future agent in the same prediction market or some agent in a separate, walled-off prediction market as the reference agent; in both cases, the reference agent can access $k$ informational substitutes that agent $t$ cannot. We defer a discussion of the details of the prediction market design to the next section, but this reasoning motivates the extension of our result from Theorem \ref{thm:bound} on cross-entropy  scoring rules to cross-entropy \emph{market} scoring rules, where agent $t$ is paid based on how much their forecast ${\bf q}^{(t)}$ improves on the market prior (assumed to be the true prior ${\bf y}^{(t-1)}$, arrived at by all previous agents reporting truthfully). 

First, identify the exact number of informational substitutes $k$ that the reference agent needs access to in order to achieve strict truthfulness, based on agent $t$'s $\tau-$granularity:
\begin{theorem}\label{thm:exactmsr}
    {Suppose agent $t$ is paid based on a cross-entropy market scoring rule with respect to agent $r$ who reports truthfully. Then, agent $t$ strictly maximizes their expected payoff by reporting truthfully if the reference agent has at least $k$ signals where}
    \begin{equation}\small
        k  > \frac{1}{-\log \left( 1-\delta \right)}\log\left(\frac{\left|\log \left(\frac{1-\eta}{\eta} \right)\right|\left({{\frac{1-\eta}{\eta}}} - {{\frac{\eta}{1-\eta}}} \right)}{8 \left(\tau\eta(1-\eta)\right)^2} \right)
    \end{equation}
\end{theorem}

Now, we give an upper bound on the gain in expected payoff by misreporting signals that does not depend on knowing $\tau-$granularity, for cases when the mechanism designer does not know $\tau$:

\begin{theorem}[Upper Bound in Deviating from Truthful Reporting under Cross-Entropy Market Scoring Rules]
        If agent $t$ is paid based on a cross-entropy market scoring rule with respect to agent $r$ who reports truthfully, the difference in expected payoff when deviating from truthful report ${\bf p}^{(t)}$ associated with signal $x^{(t)}$ to misreport ${\bf q}^{(t)}$ is:
    \begin{equation}\small \label{eq:uboundcemsr}
       \mathbb{E}[S_{CEM}({\bf q}^{(t)}) | x^{(t)}, x^{(1:t-1)}] - \mathbb{E}[S_{CEM}({\bf p}^{(t)}) | x^{(t)}, x^{(1:t-1)}] \leq \hat{\mathcal{D}}_{\eta}(\Delta, {\bf y}^{(t-1)})
    \end{equation}
    where ${\bf y}^{(t-1)}$ is the true market prior before agent $t$'s report and
    \begin{equation}\small \label{eq:dhat}
       \hat{\mathcal{D}}_{\eta}(\Delta, {\bf y}^{(t-1)}) = \Delta \cdot \log \left(\frac{({1-\eta}) +  \Delta\cdot\left( \eta\frac{y_0^{(t-1)}}{y_1^{(t-1)}} + ({1-\eta})\right)}{\eta -  \Delta\cdot \left(\eta + ({ 1-\eta}) \frac{y_1^{(t-1)}}{y_0^{(t-1)}}\right)} \right)
    \end{equation}
    \label{thm:pmcemsr}
\end{theorem}

\begin{remark}
    As $k$ increases and $ \Delta$ vanishes (by Theorem \ref{eq:imptheorem}), the difference in expected rewards from misreporting also vanishes, i.e., $\lim_{ \Delta \to 0} \hat{\mathcal{D}}_\eta(\Delta, {\bf y}^{(t-1)}) = 0$. Additionally, agent $t$'s expected reward tends as $\lim_{ \Delta \to 0} \mathbb{E}[S_{CEM}({\bf q}^{(t)}) | x^{(t)}, x^{(1:t-1)}] =-H({\bf p}^{(t)}, {\bf q}^{(t)}) + H({\bf p}^{(t)}, {\bf y}^{(t-1)})$ which is strictly maximized under truthful reporting ${\bf q}^{(t)} = {\bf p}^{(t)}$.
\end{remark}

\begin{remark} \label{rem:kproc}
 Since $\hat{\mathcal{D}}_\eta(\Delta, {\bf y}^{(t-1)}) = \varepsilon$ is a transcendental equation, we can only numerically compute ${\varepsilon'} = \min \{|\Delta|: \hat{\mathcal{D}}_{\eta}(\Delta, {\bf y}^{(t-1)}) = \varepsilon\}$. We can subsequently use this to compute the minimum number of informational substitutes $k_{min}$ (using Equation \ref{eq:kchoice}) that agent $r$ requires so the gain in deviating from truthful reporting is no more than $\varepsilon$ (according to Theorem \ref{eq:imptheorem}). We visualize the dependence of $k$ on various settings of $\delta, \eta$ and choice of $\varepsilon$ in Figure \ref{fig:main} and Appendix \ref{app:viz}, observing that $k_{min}$ is much more sensitive to the choice of $\delta$ than $\varepsilon$ or $\eta$.
\end{remark}

\begin{figure}[htbp]
    \centering
    \includegraphics[width=\textwidth]{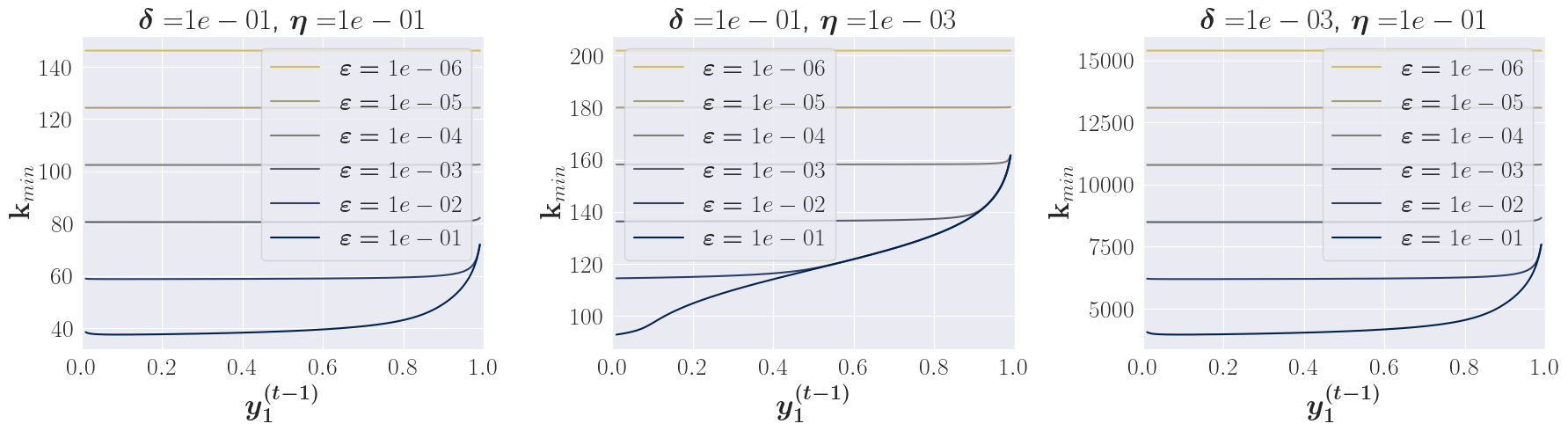}
    \caption{Minimum number of informational substitutes $k_{min}$ that the reference agent needs to guarantee that the incentive to deviate from truthful reporting is no more than $\varepsilon$ as a function of the prior market prediction $y_1^{(t-1)}$, for different choices of $(\delta, \eta)$ when $\tau-$granularity is not known. }
    \label{fig:main}
\end{figure}

\section{Self-Resolving Prediction Markets}

In the previous section, we saw that if we use cross-entropy market scoring rules to score an agent $t$ against a reference agent with sufficient informational substitutes, agent $t$ strictly maximizes their payoff by reporting truthfully. In this section, we use this insight to propose our design for a \emph{self-resolving prediction market} that does not need access to the true outcome. We show that it is a perfect Bayesian equilibrium for all agents to report their true beliefs and believe that all other agents report truthfully. In addition to eliciting information from agents, the prediction market design also naturally aggregates the predictions to provide a single prediction summarizing all elicited information. There are several possible ways to design such a self-resolving prediction market, and we begin with our preferred design.

\paragraph{Our Preferred Design} In our preferred design, the prediction market elicits predictions from agents arriving in sequence, terminates with some probability $\alpha$ after every report, and pays all but the final $k$ agents using cross-entropy market scoring rules with the terminal agent serving as the reference agent $r$. The final $k$ agents are simply paid a flat fee $R$. The primary incentive issue is that agent $t$ seeks to both closely match their reference prediction and also to nudge their reference agent towards more certain predictions. However, as long as agent $t$'s reference agent $r$ satisfies $r \geq t+k$ where $k$ is chosen as in Remark \ref{rem:kproc}\footnote{Technically, setting $k$ using Remark \ref{rem:kproc} depends on the market prior ${\bf y}^{(t-1)}$ which varies for each agent and time step. However, as seen in Figure \ref{fig:main}, the variation is not large, and is even flat for small enough $\varepsilon$. Thus, for any chosen $\varepsilon$ we can just choose $k = \max_{{\bf y}^{(t-1)} \in [y_{min}, y_{max}]} k_{min}({\bf y}^{(t-1)})$ for some $y_{min}, y_{max}$ a very small amount away from 0 and 1 respectively to avoid any pathologies with the boundaries. $k$ can also be set without reference to the prior according to Theorem \ref{thm:exactmsr}.} or Theorem \ref{thm:exactmsr}, truthful reporting is either $\varepsilon$-incentive-compatible or strictly incentive-compatible. With such a design, market predictions prior to $t$ are shared signals between agents $t$ and $r$, i.e., $x^{(s)} = \tilde{x}^{(1:t-1)}$. 
Agent $t$'s private information $x^{(t)}$ remains private unless truthfully revealed, and agent $r$'s `private signals' that agent $t$ cannot access are comprised of their own signal and predictions made after agent $t$, $\hat{x}^{(r)} = \{x^{(r)}, {\bf q}^{(t+1: r)}\}$.  We discuss other possible designs in Section \ref{sec:intprac}. In this case, the cost to the mechanism designer is $k\cdot R + \sum_{t=1}^{r-k} -H({\bf r}, {\bf q}^{(t)}) + H({\bf r}, {\bf q}^{(t-1)}) = k\cdot R -H({\bf r}, {\bf q}^{(r-k)}) + H({\bf r}, {\bf q}^{(0)})$.

\begin{theorem}[Truthful Reporting is a Perfect Bayesian Equilibrium]\label{thm:rectru}
    In a self-resolving prediction market where a reference agent can observe at least $k$ additional predictions (with $k$ chosen according to Theorem \ref{thm:exactmsr}), the strategy profile of truthful reporting and belief profile of believing all previous agents reported truthfully is a strict Perfect Bayesian Equilibrium. When $k$ is chosen according to Remark \ref{rem:kproc}, this strategy and belief profile is a $\varepsilon-$PBE. Off the equilibrium path, agents believe any observed (mis)report to be truthful if that report is consistent with the reporting agent's signal structure. If the observed prediction is invalid, (i.e., inconsistent with any signal the reporting agent could have observed), all subsequent agents may either update as if they learned some arbitrary signal or ignore this prediction, and believe all subsequent agents continue to report truthfully. 
 \end{theorem}

Our self-resolving prediction market design also admits some less plausible PBE, like uninformative equilibria (where all agents report the same prediction regardless of true belief) and permutation equilibria (where all agents re-label their beliefs on $Y$). Intuitively, if every agent reports the same prediction, they do not improve on a peer's report, and hence receive zero payout -- strictly less than in the informative equilibrium. We discuss this in greater detail in Appendix \ref{app:equil}.

\subsection{Market Termination with Random Stopping} \label{app:term}

Next, we expand on the mechanics of our preferred self-resolving prediction market design, which terminates with some probability $\alpha$ and sets the final agent as the fixed reference agent, paying all but the last $k$ agents with this reference agent. The last $k$ agents are simply paid some constant reward $R$. This guarantees that we only score an agent against the final agent if they are at least $k$ steps removed. In such a design, every agent $t$ receives the constant reward $R$ with probability $\sum_{n=1}^{k} (1-\alpha)^{n-1}\alpha = 1- (1-\alpha)^k$, and is paid according to the cross-entropy market scoring rule with probability $(1-\alpha)^k$. When $k$ is chosen according to Theorem \ref{thm:exactmsr}, we maintain strict truthfulness. When $k$ is chosen according to Equation \ref{eq:kchoice}, this gives us $\left((1-\alpha)^k\cdot \varepsilon\right)-$incentive-compatibility; by $\mathbb{E}[S_{CEM}({\bf q}^{(t)}) | \hat{X}^{(t)}] \leq \mathbb{E}[S_{CEM}({\bf p}^{(t)}) | \hat{X}^{(t)}] + \varepsilon$, we can write any agent $t$'s expected utility as:
    
    \begin{equation}
       \begin{split}
           \mathbb{E}[U({\bf q}^{(t)})] &= \left(\sum_{n=1}^k (1-\alpha)^{n-1} \alpha R\right) + (1-\alpha)^k \cdot \mathbb{E}[S_{CEM}({\bf q}^{(t)}) | \hat{X}^{(t)}] \\
           &\leq \left(1 - (1-\alpha)^k\right) R + (1-\alpha)^k \left(\mathbb{E}[S_{CEM}({\bf p}^{(t)}) | \hat{X}^{(t)}] + \varepsilon\right) \\
           &= \mathbb{E}[U({\bf p}^{(t)})] + (1-\alpha)^k\varepsilon
       \end{split}
    \end{equation}

\paragraph{Choice of $\alpha$} For risk-neutral agents, these choices do not matter in theory since any agent is no more than $\varepsilon$ better off by misreporting for any choice of $R$ or $\alpha$. Yet, as a practical matter, we would like to choose an $\alpha$ not so large that the market typically terminates quickly and pays most agents flat rewards, and choose $\alpha$ not so small that the market does not terminate quickly enough. Thus, the mechanism designer may choose some desired expected number of traders $T$. Then, by setting $\alpha = \frac1{T+k}$, the mechanism will average $T+k$ participating agents before termination, where $T$ agents can be paid according to cross-entropy market scoring rules and $k$ traders can be paid a flat fee.

\paragraph{Random Stopping without Flat Fees} Lastly, we observe that it is also possible to pay every agent, including the last $k$ agents, based on the final (reference) agent's prediction this could be done by finding the setting of $\alpha$ such that the expected gain from deviating in a supposedly truthful PBE (while allowing for the reference agent to be less than some $k$ timesteps away) is no more than $\varepsilon$, i.e., solving $\mathbb{E}_{\alpha}[\hat{\mathcal{D}}_\eta(\Delta, {\bf p}^{(t-1)})] \leq \varepsilon$ for $\alpha$. To compute this, we would first use the upper bound on $\Delta$ in Theorem \ref{eq:imptheorem} to arrive at an upper bound on $\hat{\mathcal{D}}_\eta$ for each possible $k$, and then weight this upper bound by $(1-\alpha)^{k-1}\alpha$ to take the expectation, and solve for $\alpha$. This computation would again need to be done numerically. The idea is that from each trader's perspective, they expect the market to continue for long enough that in expectation, their reference agents will be sufficiently `distant' so they can expect to gain no more than $\varepsilon$ by deviating from truthful reporting.

\subsection{Other Design Choices} \label{sec:intprac}

Here, we consider some alternative design decisions, such as running a parallel prediction market for reference, ideas for terminating self-resolving prediction markets, strategies for choosing the reference agent, and the cost to the mechanism designer.

\paragraph{Parallel Prediction Markets} Instead of choosing a reference agent from the same prediction market, we could also run one or more separate, walled-off \emph{parallel prediction market} on the same question and choose agents in one market as the reference for agents in the other using cross-entropy MSRs. Importantly, we must ensure that agents in one market cannot observe the predictions in the other(s). Agent $r$ can neither observe agent $t$'s report nor do the two agents share any signals; rather each agent only has access to their own private signal and previous predictions in their own market. In doing so, we trade the incentive to strategically mislead reference agents for an incentive to hedge towards the common prior in light of uncertainty about the reference agent's information and beliefs. Yet, if we numerically choose $k$ appropriately as in Theorem \ref{thm:exact} (or Remark \ref{rem:kproc}), agents cannot gain in expectation (or gain more than $\varepsilon$) by misreporting their posterior. 
Ultimately, there is no deep theoretical difference in incentives between our preferred design and parallel prediction markets; if the reference agent has enough informational substitutes, we will be able to guarantee that the incentive to deviate from truthful reporting vanishes. However, operating two different prediction markets is inefficient as there are two `prices' each aggregating only half the information; thus, we favour the design where the reference agent is chosen to be the terminal agent in the same market.

\paragraph{Strategies for Choosing Reference Agent} While we choose the terminal agent as the reference agent which requires agents to wait for market termination, there are several other options given a required $k$: (1) use a {\bf rolling window} of size $k$, so agents receive can receive their payouts relatively quickly without having to wait for the market to terminate (however, this has other downsides, as discussed next); (2) select a {\bf reference agent per batch}, i.e., split agents into batches that share the same reference agent, which has the benefit of quick payoffs for agents and infrequent payoff determinations for the mechanism designer. Additionally, the mechanism can (and probably should) average the predictions of some number of reference agents to decrease the variance in payoffs, especially in the presence of uninformed agents or `noise traders.' 

\paragraph{Market Termination Options}  Self-resolving prediction markets can be run indefinitely if reference agents are chosen periodically to payoff previous agents. Alternatively, if the mechanism designer has access to $k$ trusted agents, these can be placed as the final $k$ agents in the market before termination. Incentive-compatibility then follows from backward induction. In general, the most plausible approach is our preferred design of simply terminating the market with some small probability $\alpha = \frac1{T+k}$ (where $T$ is the desired average number of traders to elicit predictions from) and paying a flat fee to the last $k$ agents who do not have a reference agent as discussed above.

\paragraph{Cost to the Mechanism Designer} The cost to the mechanism designer is simply the sum of the payouts to each agent and the constant payouts to agents without reference agents. 
With a rolling window of size $k$ or batched agents, the worst case cost to mechanism designer is unbounded, since agents are being scored against changing reference predictions. However, using fixed reference agents (such as the final agent in the market) can ensure bounded loss, which motivates our preferred design.

\subsection{Trading Securities with a Cost Function-Based Automated Market Maker} \label{app:cfmm}

Here, we show that self-resolving prediction markets with fixed reference agents and cross-entropy market scoring rules can easily be implemented as a cost function-based market maker, allowing agents to \emph{trade} securities that payout upon expiration -- a more familiar design of prediction markets. By virtue of being built on the log-scoring rule, the automated market maker (AMM) for cross-entropy market scoring rule is nearly identical to the AMM for the log market scoring rule \citep{chen2012utility}. The only difference is that in standard prediction markets, contracts corresponding to the event that happened pay out $\$1$, while contracts corresponding to any event that did not happen expire worthless. In our self-resolving prediction markets, contracts pay out at the probability set by the reference agent, i.e., contracts on $Y=0$ pay out $r_0$ cents, and contracts on $Y=1$ pay out $r_1$ cents.

\paragraph{Cost function-based Market Maker} Consider an automated market maker with a quantity vector of contracts sold so far ${\bf c} = \begin{pmatrix} c_0 & c_1 \end{pmatrix}$ where $c_i$ corresponds to the number of contracts sold corresponding to outcome $Y_i$. Let the AMM have a cost-function $C$ that maps the quantity vector to the total cash raised through the sale (and purchase) of securities as $C({\bf c})$. When agent $t$ wants to buy or sell contracts, the going price $q_i$ for security $i$ is the instantaneous change in `cash raised' or the cost function, i.e., $q_i({\bf c}) = \nicefrac{\partial C}{\partial c_i}$. In buying or selling contracts, agent $t$ changes the quantity vector from ${\bf c}^{(t-1)}$ to ${\bf c}^{(t)}$, and pays in (or receives, if negative) $C({\bf c}^{(t)}) - C({\bf c}^{(t-1)}) = \int_{{\bf c}^{(t-1)}}^{{\bf c}^{(t)}} w({\bf c}) \,d{\bf c}$.  When the market terminates with reference prediction ${\bf r}$, $c_0$ contracts expire at $r_0$ cents and $c_1$ contracts expire at $r_1$ cents, so the AMM must pay out $c_0r_0 + c_1r_1$, having raised $C({\bf c})$. 

\paragraph{Equivalence to Cross-Entropy Market Scoring Rules} To establish an equivalence between our cross-entropy market scoring rules and a cost function-based automated market maker, we require that $S_\text{CEM}({\bf r}, {\bf q}^{(t)}, {\bf q}^{(t-1)}) = S_\text{CE}({\bf r}, {\bf q}^{(t)})  - S_\text{CE}({\bf r}, {\bf q}^{(t-1)}) = ({\bf c}^{(t)} - {\bf c}^{(t-1)})\cdot {\bf r} - \left(C({\bf c}^{(t)}) - C({\bf c}^{(t-1)}\right)$. Thus, we arrive at a similar system of three equations as \citet{chen2012utility}: (1) $S_{\text{CE}}({\bf r}, {\bf q}) = {\bf c}\cdot {\bf r} - C({\bf c})$, (2) $\sum_i q_i = 1$ (i.e., prices across outcomes sum to 1 to prevent arbitrage), and (3) $q_i = \nicefrac{\partial C}{\partial c_i}$. The second and third equations are exactly the same as the system of equations given by \citep{chen2012utility}, while the first is a slight modification that ensures the payout depends on the reference prediction. 

Observe that we can use the same cost-function as the log scoring rule, $C({\bf c}) = b \cdot \log \left(\sum_j e^{\nicefrac{c_j}{b}}\right)$ with price function $q_i = \frac{e^{\nicefrac{c_i}{b}}}{\sum_j e^{\nicefrac{c_j}{b}}}$ satisfying conditions (2) and (3) ($b$ is a liquidity parameter). Then, equation (1) is also satisfied:

\begin{equation}\small
    \begin{split}
         S_\text{CE}({\bf r}, {\bf q}^{(t)}) &= b \left( r_1 \log(q_1) + r_0 \log(q_0) \right) \\
         &= b\left( r_1 \log\left( \frac{e^{\nicefrac{c_1}{b}}}{e^{\nicefrac{c_0}{b}} + e^{\nicefrac{c_1}{b}}}\right) + r_0 \log\left( \frac{e^{\nicefrac{c_0}{b}}}{e^{\nicefrac{c_0}{b}} + e^{\nicefrac{c_1}{b}}}\right) \right) \\
         &= b\left( r_1 \log\left( {e^{\nicefrac{c_1}{b}}}\right) - r_1\log\left({e^{\nicefrac{c_0}{b}} + e^{\nicefrac{c_1}{b}}}\right) + r_0 \log\left( {e^{\nicefrac{c_0}{b}}}\right) - r_0\log\left({e^{\nicefrac{c_0}{b}} + e^{\nicefrac{c_1}{b}}}\right) \right) \\
         &= b\left( r_1 {\frac{c_1}{b}} + r_0 {\frac{c_0}{b}} - (r_1 + r_0)\log\left({e^{\nicefrac{c_0}{b}} + e^{\nicefrac{c_1}{b}}}\right)  \right) \\
         &= c_1r_1 + c_0 r_0 - b \cdot \log \left( e^{\nicefrac{c_0}{b}} + e^{\nicefrac{c_1}{b}} \right) \\
         &= {\bf c}\cdot {\bf r} - C({\bf c})
    \end{split}
\end{equation}

Thus, our self-resolving prediction market using a fixed reference agent where agents report probabilistic predictions and are paid based on cross-entropy market scoring rules can equivalently be implemented as agents trading contracts through a cost function-based market maker where contracts payout at prices determined by the reference agent's trade.

\section{Conclusion}

We introduced the first incentive-compatible design for self-resolving prediction markets that can efficiently elicit and aggregate information in the absence of ground truth. Our mechanism sequentially elicits predictions from agents, randomly terminates with probability $\alpha$, and pays all but the final $k$ agents based on the cross-entropy between their prediction and the terminal agent, with the final $k$ agents receiving a flat fee. We showed that truthful reporting is a strict perfect Bayesian equilibrium, and that the payout in uninformative equilibria is zero. Our key insight was that while a simple `peer prediction' scheme of paying agents based on agreement with a peer is not truthful, it becomes truthful when the reference agent has access to enough private conditionally independent signals (i.e., informational substitutes). 

Our analysis of self-resolving prediction markets (or equivalently, sequential peer prediction) opens up rich directions for future work. One important direction is to consider situations where our informational substitutes assumption does not hold, e.g., when agents have a cost to exerting effort or acquiring information. It may be that agent signals are not conditionally dependent given the outcome $Y$, and are perhaps only conditionally independent given the outcome, agent effort, and agent expertise. Thus, studying how agents can signal their effort or expertise in a simple way is an important direction for future work. It would also be interesting to explore incentives in the presence of risk aversion. Lastly and most importantly, empirically tests of our proposed mechanism would shed further light on the viability of our mechanism in aggregating agents' beliefs on both resolvable and unresolvable questions ranging from causal effects, to long-term forecasts about economic indicators, to data-collection tasks widely used in machine learning.

\newpage 

\bibliographystyle{apalike}
\bibliography{biblio}

\newpage 

\appendix
\section{Example of Incentive to Strategically Mislead under Cross-Entropy Scoring Rules} \label{sec:exinct}
As a concrete example, consider the following information structure. Suppose the market's initial probability for $Y=1$ is 0.5, and say agent 1 has a partially informative private signal, in that if they receive private signal $X^{(1)}=0$, their posterior over $Y=1$ is 0.49, and if they receive private signal $X^{(1)}=1$ their posterior over $Y=1$ is 0.99. Further suppose that the subsequent agent 2 only has a very weakly informative private signal, such that if they see $X^{(2)}=0$, they multiply the market odds by 0.99, and if they see $X^{(2)}=1$, they multiply the market odds by $1.01$. Now, let's say agent 1 receives $X^{(1)}=0$ and updates to a posterior of 0.49. If agent 1 reports their true posterior of 0.49 to the mechanism, agent 2's posterior is either 0.487 or 0.492 depending on whether their private signal is $X^{(2)} = 0$ or $X^{(2)} = 1$ respectively. This nets agent 1 a payoff of approximately $-0.693$ either way. Alternatively, if agent 1 \emph{misreports} their posterior to be 0.99 \emph{as if they had received the signal} $X=1$, agent 2 believing this to be truthful will report a posterior of about 0.99 as well, netting agent 1 a significantly higher reward of approximately $-0.055$. Clearly, at any supposed truthful equilibrium, agent 1 could have an incentive to deviate from truthful reporting.

\section{Visualizing Incentives} \label{app:viz}

Here, we visualize incentives to deviate based on the number of signals the reference agent has. This analysis does not consider the $\tau-$granularity of agents' signal spaces.

\subsection{Visualizing Incentives under Cross-Entropy Scoring Rules} \label{sec:vicesr}

We made three observations in Section \ref{sec:ppcesr} regarding the upper bound on the gain in deviating from truthful reporting $\mathcal{D}_\eta(\Delta, {\bf y)})$ from Equation \ref{eq:deta} in Theorem \ref{thm:bound}, which we enumerate and visualize here:

\begin{enumerate}
    \item \textbf{Upper bound on gain from deviating from truthful reporting $\mathcal{D}_\eta(\Delta, {\bf y})$ vanishes with the adjustment term $\Delta$}: As the adjustment term $\Delta \to 0$, the upper bound on the incentive to deviate from truthful reporting also goes to zero, i.e., $\mathcal{D}_\eta(\Delta, {\bf y}) \to 0$. We plot this upper bound $\mathcal{D}_\eta(\Delta, {\bf y})$ as a function of the adjustment term $\Delta$ for various choices of ${\bf y}$ and $\eta$ in Figure \ref{fig:mudel}. 
    
    One interesting observation is that the curve appears `cut off' for various values of $\Delta$. The reason for this is that $\Delta$, as defined in Lemma \ref{lem:post}, is an adjustment to a valid posterior ${\bf p}^{(t)}$ that must remain a valid probability distribution. This requirement essentially translates to a requirement that the denominator in Equation \ref{eq:deta}  be greater than zero (see proof of Theorem \ref{eq:imptheorem}), i.e., $\eta y_0 -  \Delta\cdot \left(\eta y_0 + ({1-\eta}){y_1}\right)$. Equivalently, we have $ \Delta <  \frac{\eta y_0}{\left(\eta y_0 + ({1-\eta}){y_1}\right)}$, a natural upper bound positive upper bound on $\Delta$. We similarly have a lower bound on $\Delta$ from the requirement that the numerator be greater than 0, i.e., $({1-\eta})  y_1 +  \Delta \cdot \left(\eta y_0 + ({1-\eta}) {y_1}\right) > 0 \Rightarrow \Delta > - \frac{({1-\eta})  y_1}{\left(\eta y_0 + ({1-\eta}) {y_1}\right)}$. Thus, when we defined ${\varepsilon'} = \min \{|\Delta|: \mathcal{D}_{\eta}(\Delta, {\bf y}) = \varepsilon\}$, we can observe that $\varepsilon'$ will be the positive-valued solution when $y_1 > \eta$ and the negative-valued solution when $y_1 < \eta$.

    \item \textbf{$\varepsilon'$, the maximum allowed magnitude of the adjustment term $\Delta$ to ensure that the gain in deviating from truthful reporting is no more than $\varepsilon$, falls gracefully with $\varepsilon$}: In Figure \ref{fig:maxdel}, we investigate what any desired upper bound on the incentive to deviate ($\varepsilon$) implies for the maximum allowed adjustment to the agent's true posterior, where $\varepsilon' = \min \{|\Delta|: \mathcal{D}_{\eta}(\Delta, {\bf y}) = \varepsilon\}$. Indeed, we find that an order of magnitude decrease in $\varepsilon$ corresponds roughly to an order of magnitude decrease in $\varepsilon'$. Additionally, $\varepsilon'$ dependence on the prior $y_1$ is relatively well-behaved across a broad range of $y_1$; however, note that at the edges (not visualized) where $y_1$ is very close to 1 or 0, the curve can turn quite steeply. Yet, if the prior is already so concentrated, in practical terms there is not much need to elicit forecasts.

\item \textbf{$k_{min}$, the minimum number of private informational substitutes that the reference agent must have that agent $t$ does not have, scales gracefully with $\varepsilon$}: In Figure \ref{fig:pk}, we plot the minimum number of informational substitutes that the reference agent must have (that agent $t$ does not) in order to guarantee that agent $t$'s maximum gain from deviating is no more than $\varepsilon$. We arrive at this for any desired $\varepsilon$ by first computing the upper bound $\varepsilon'$ on the adjustment term $\Delta$ as described above, and then using this to compute $k$ through Equation \ref{eq:kchoice} of Theorem \ref{eq:imptheorem}.  We find that the $k$ is significantly more sensitive to the setting of $\delta$ than the setting of $\eta$ or the choice of $\varepsilon$.

\end{enumerate}

\begin{figure}[h!]
    \centering
    \includegraphics[width=\textwidth]{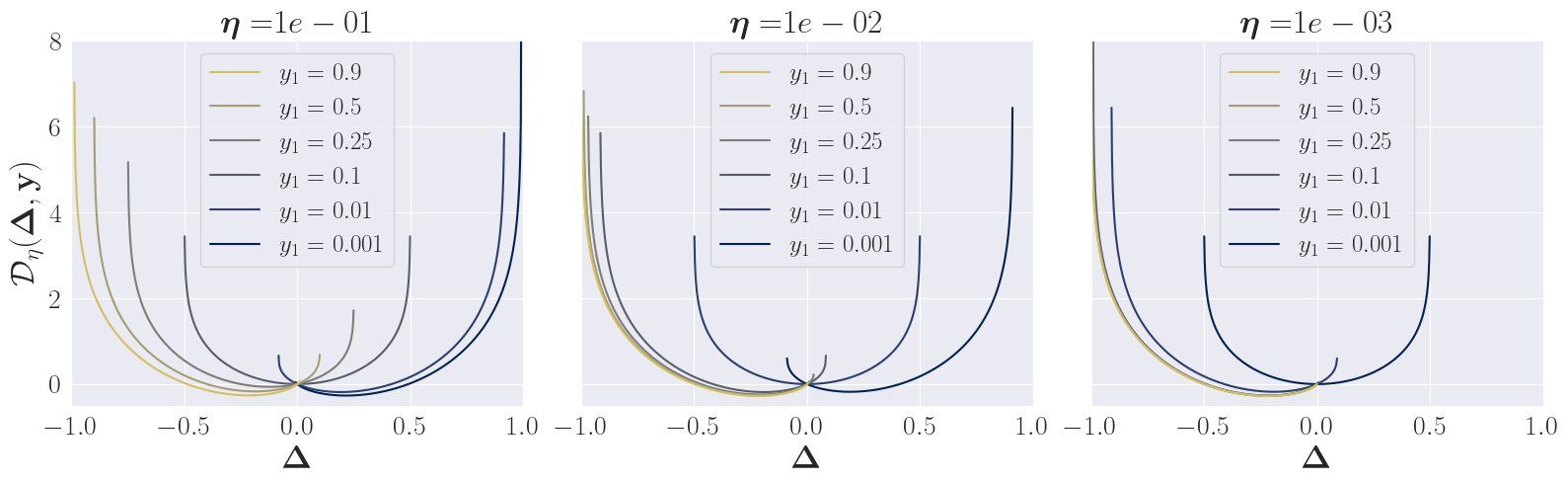}
    \caption{(Cross-Entropy Scoring Rules) Plot of the upper bound of gain in deviating from truthfulness as a function of the adjustment term $\Delta$. Observe that as the adjustment term $\Delta \to 0$, the upper bound $\mathcal{D}_\eta(\Delta, {\bf y}) \to 0$ as well, showing that the incentive to deviate can be made to disappear by sending the adjustment term to zero.}
    \label{fig:mudel}
\end{figure}

 \begin{figure}[htbp]
    \centering
    \includegraphics[width=\textwidth]{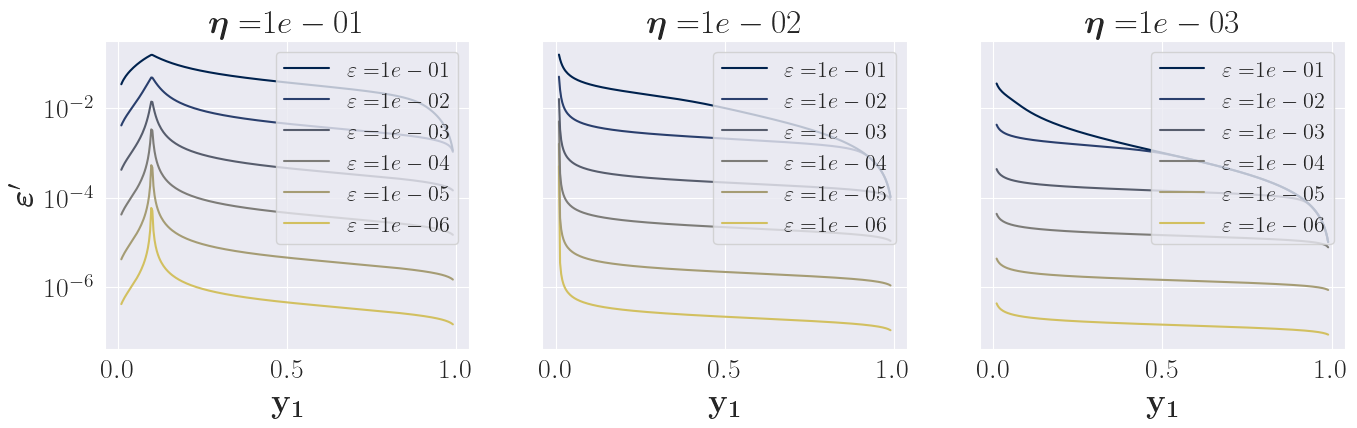}
    \caption{(Cross-Entropy Scoring Rules)  Plot of maximum allowed adjustment $\varepsilon'$ to guarantee that the incentive to deviate from truthful reporting is no more than $\varepsilon$ as a function of the prior on $Y=1$ as $y_1$, for various choices of $\varepsilon$.}
    \label{fig:maxdel}
\end{figure}

\begin{figure}[h!]
    \centering
    \includegraphics[width=\textwidth]{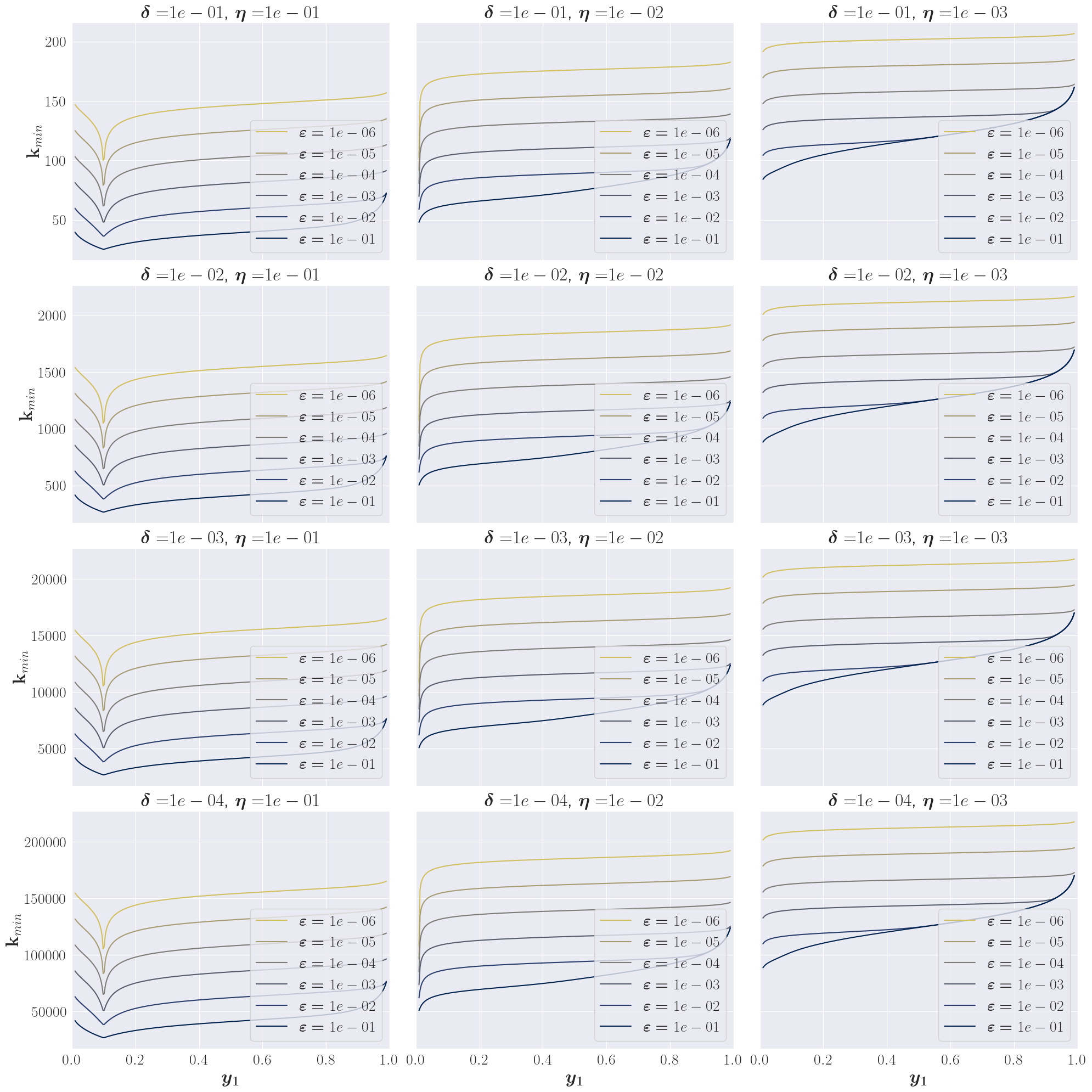}
    \caption{(Cross-Entropy Scoring Rules)  Plot of minimum number of information substitutes that the reference agent needs to guarantee that the incentive to deviate from truthful reporting is no more than $\varepsilon$ as a function of the prior $y_1$ under cross-entropy scoring rules, for different choices of $(\delta, \eta)$. }
    \label{fig:pk}
\end{figure}

\newpage

\subsection{Visualizing Incentives under Cross-Entropy Market Scoring Rules} 

Here, we carry out a similar visualization as above, but for cross-entropy \emph{market} scoring rules, where the upper bound on the incentive to deviate $\hat{\mathcal{D}}_\eta({\Delta}, {\bf y}^{(t-1)})$ is given by Equation \ref{eq:dhat} in Theorem \ref{thm:pmcemsr}.

\begin{enumerate}
    \item \textbf{Upper bound on gain from deviating from truthful reporting $\hat{\mathcal{D}}_\eta(\Delta, {\bf y}^{(t-1)})$ vanishes with the adjustment term $\Delta$}: As the adjustment term $\Delta \to 0$, the upper bound on the incentive to deviate from truthful reporting also goes to zero, i.e., $\hat{\mathcal{D}}_\eta(\Delta, {\bf y}^{(t-1)}) \to 0$. We plot this upper bound $\hat{\mathcal{D}}_\eta(\Delta, {\bf y}^{(t-1)})$ as a function of the adjustment term $\Delta$ for various choices of ${\bf y}$ and $\eta$ in Figure \ref{fig:mudel2}.      Once again, the curve appears `cut off' for various values of $\Delta$. As before, this stems from the fact that $\Delta$, as defined in Lemma \ref{lem:post}, is an adjustment to a valid posterior ${\bf p}^{(t)}$ that must remain a valid probability distribution. This requirement thus translates to a requirement that the numerator and denominator in Equation \ref{eq:dhat} be greater than zero (see proof of Theorem \ref{thm:pmcemsr}), i.e., $\eta -  \Delta\cdot \left(\eta + ({ 1-\eta}) \frac{y_1^{(t-1)}}{y_0^{(t-1)}}\right) > 0\Rightarrow \Delta <  \frac{\eta}{ \eta + ({ 1-\eta}) \frac{y_1^{(t-1)}}{y_0^{(t-1)}}}$, and $({1-\eta}) +  \Delta\cdot\left( \eta\frac{y_0^{(t-1)}}{y_1^{(t-1)}} + ({1-\eta})\right) > 0 \Rightarrow \Delta > -\frac{{1-\eta}}{   \eta\frac{y_0^{(t-1)}}{y_1^{(t-1)}} + ({1-\eta})}$. The curves in Figure \ref{fig:mudel2} get cut off as $\Delta$ approaches these bounds.

    \item \textbf{$\varepsilon'$, the maximum allowed magnitude of the adjustment term $\Delta$ to ensure that the gain in deviating from truthful reporting is no more than $\varepsilon$, falls gracefully with $\varepsilon$}: In Figure \ref{fig:maxdel2}, we investigate what any desired upper bound on the incentive to deviate ($\varepsilon$) implies for the maximum allowed adjustment to the agent's true posterior, where $\varepsilon' = \min \{|\Delta|: \hat{\mathcal{D}}_{\eta}(\Delta, {\bf y}^{(t-1)}) = \varepsilon\}$. Indeed, we find that an order of magnitude decrease in $\varepsilon$ corresponds roughly to an order of magnitude decrease in $\varepsilon'$. Additionally, $\varepsilon'$ dependence on the prior $y_1^{(t-1)}$ is relatively well-behaved across a broad range of $y_1^{(t-1)}$. Note that the $y$-axis in Figure 6 is on a log scale, so the variation in $y_1^{(t-1)}$ within the same order of magnitude is masked by the plot. Additionally, when the prior $y_1^{(t-1)}$ is very close 1 or 0, the curve can turn steeply; but if the prior is so strong,  there is not much practical need to elicit forecasts.

\item \textbf{$k_{min}$, the minimum number of private informational substitutes that the reference agent must have that agent $t$ does not have, scales gracefully with $\varepsilon$}: In Figure \ref{fig:pk2}, we plot the minimum number of informational substitutes that the reference agent must have (that agent $t$ does not) in order to guarantee that agent $t$'s maximum gain from deviating is no more than $\varepsilon$. We arrive at this for any desired $\varepsilon$ by first computing the upper bound $\varepsilon'$ on the adjustment term $\Delta$ as described above, and then using this to compute $k$ through Equation \ref{eq:kchoice} of Theorem \ref{eq:imptheorem}. We find that the setting of $\delta$ matters much more than $\eta$ or $\varepsilon$.

\end{enumerate}

\begin{figure}[h!]
    \centering
    \includegraphics[width=\textwidth]{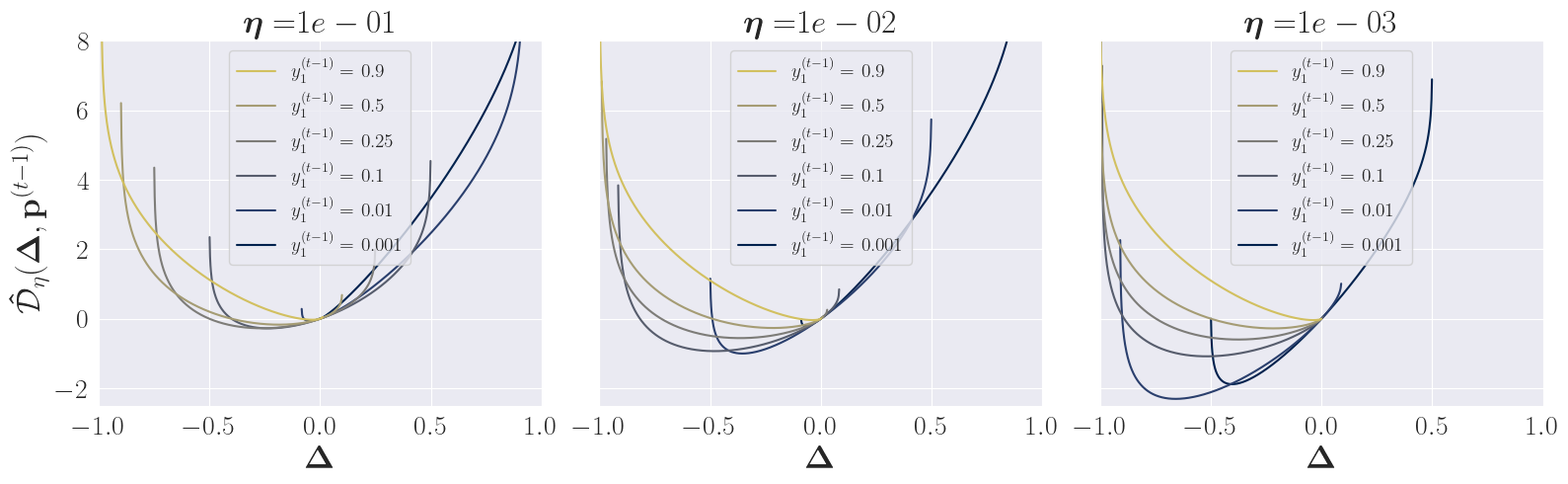}
    \caption{(Cross-Entropy Market Scoring Rules)  Plot of the upper bound of gain in deviating from truthfulness as a function of the adjustment term $\Delta$. Observe that as the adjustment term $\Delta \to 0$, the upper bound $\hat{\mathcal{D}}_\eta(\Delta, {\bf y}^{(t-1)}) \to 0$ as well, showing that the incentive to deviate can be made to disappear by sending the adjustment term to zero.}
    \label{fig:mudel2}
\end{figure}

\begin{figure}[h!]
    \centering
    \includegraphics[width=\textwidth]{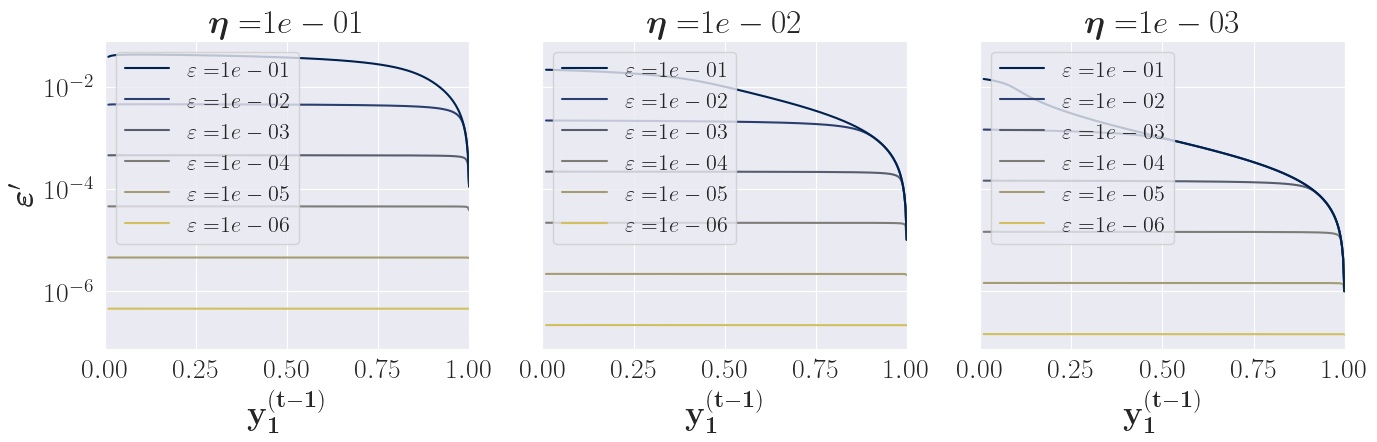}
    \caption{(Cross-Entropy Market Scoring Rules) Plot of maximum allowed adjustment $\varepsilon'$ to guarantee that the incentive to deviate from truthful reporting is no more than $\varepsilon$ as a function of the market prior on $Y=1$ as $y_1^{(t-1)}$, for various choices of $\varepsilon$. }\label{fig:maxdel2}
\end{figure}

\begin{figure}[t!]
    \centering
    \includegraphics[width=\textwidth]{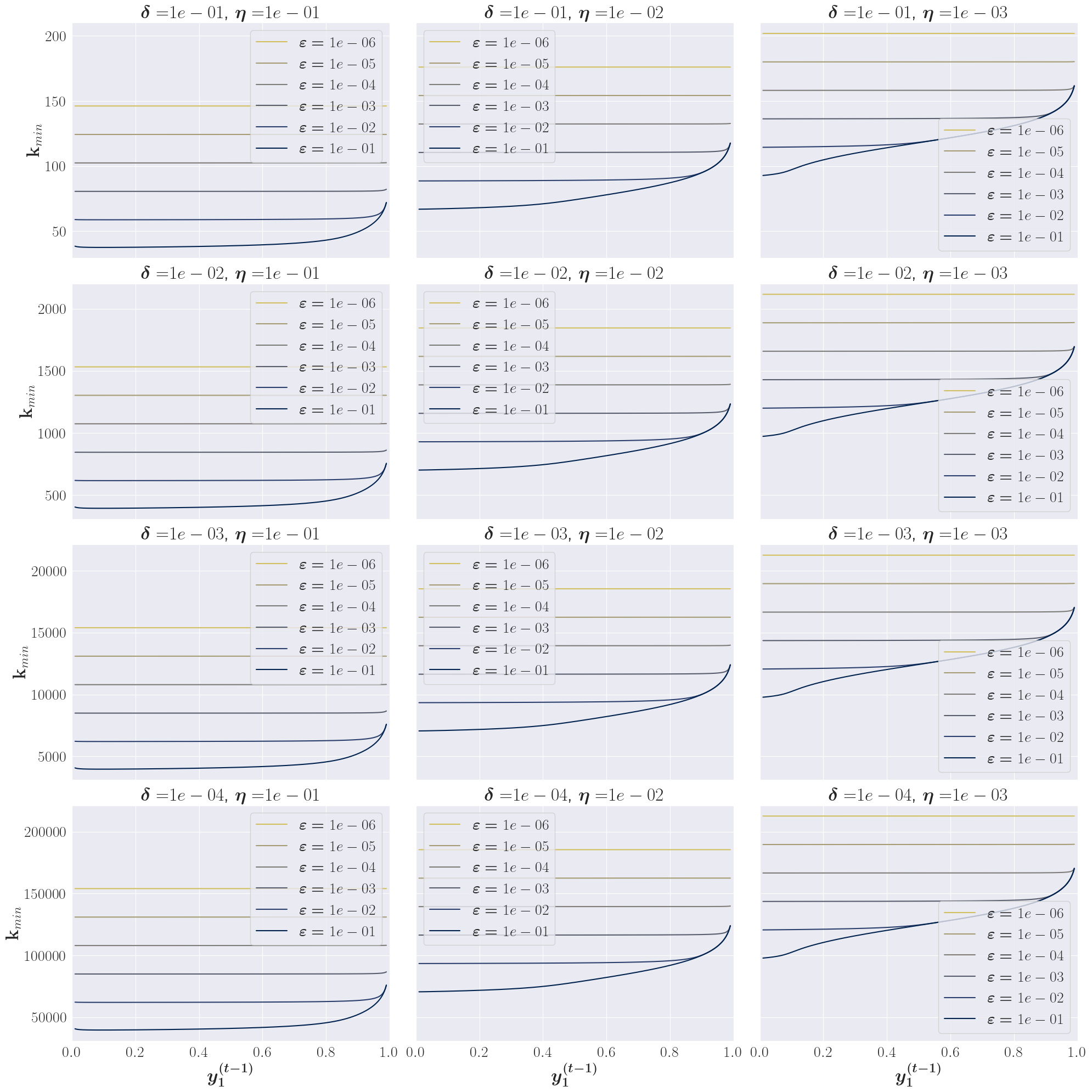}
    \caption{(Cross-Entropy Market Scoring Rules) Plot of minimum number of information substitutes that the reference agent needs to guarantee that the incentive to deviate from truthful reporting is no more than $\varepsilon$ as a function of the prior $y_1^{(t-1)}$ under cross-entropy MSR, for different choices of $(\delta, \eta)$.}
    \label{fig:pk2}
\end{figure}

\section{Other Equilibria} \label{app:equil}

\subsection{Uninformative Equilibria}

In the absence of ground truth, it is a fundamental challenge to ensure that an information elicitation mechanism can even elicit information at all. So-called uninformative equilibria -- where all agents report the same uninformative report regardless of their private signal -- has remained a challenge when designing information elicitation mechanisms without ground truth.

In our case, while we are not able to get rid of uninformative equilibria, we can essentially eliminate their appeal through the use of market scoring rules; this ensures agents' payouts are all zero, and thus less than the truthful equilibrium. The only other peer prediction mechanism to make this guarantee in the \emph{single-task} setting (i.e., without relying on correlations across multiple IID tasks) is the Differential Peer Prediction (DPP) mechanism by \citep{schoenebeck2020two}. Although their mechanism consists of two sequential stages, it does not efficiently aggregate information from all agents as we are able to do with a prediction market design. On the other hand, their result is not restricted to the case of informational substitutes.

\begin{theorem}
    In any self-resolving prediction market, the strategy profile where all agents report the same signal and any belief profile where agents believe subsequent agents will report the same uninformative signal (both on and off the equilibrium path) is a perfect Bayesian equilibrium. Additionally, the uninformative equilibrium pays zero reward to all agents (except the first) paid by the cross-entropy market scoring rule.
\end{theorem}

\begin{proof}
    For any agent $t$, suppose all previous agents $j<t$ reported prediction the same prediction that agent $t+k$ will make (i.e., ${\bf q}^{(j)} = {\bf q}^{(t+k)}$) regardless of what agent $t$ reports. Then, agent $t$'s payoff under cross-entropy market scoring rules is maximized when reporting ${\bf q}^{(t)} = {\bf q}^{(t+k)}$, and agent $t$ has no incentive to deviate from uninformative reporting. Since this holds for every agent $t$ who believes that subsequent agents report an uninformative prediction, regardless of agents' beliefs about $Y$ or past agents' reports, the uninformative strategy profile is a perfect Bayesian equilibrium. Even if a prior agent goes off the equilibrium path by reporting a different signal, as long as agents believe subsequent agents will continue to report the uninformative prediction, they will themselves be incentivized to report the uninformative prediction.

    Under cross-entropy market scoring rules with respect to reference agent $t+k$, every agent $t$'s payoff in the uninformative equilibrium is zero:
    \begin{equation}\small
    \begin{split}
        \mathbb{E}[S_{CEM}({\bf q}^{(t)}) | \hat{X}^{(t)}]
    &= -H\left({\bf q}^{(t+k)}, {\bf q}^{(t)} \right)  + H\left({\bf q}^{(t+k)}, {\bf q}^{(t-1)} \right) \\
    &= -H\left({\bf q}^{(t+k)}, {\bf q}^{(t+k)} \right)  + H\left({\bf q}^{(t+k)},{\bf q}^{(t+k)} \right) \\
    &= 0
    \end{split}
    \end{equation}

    The first agent receives a reward of $-H\left({\bf q}^{(1+k)}, {\bf q}^{(1+k)} \right)  + H\left({\bf q}^{(1+k)}, {\bf q}^{(0)} \right) = -H\left({\bf q}^{(1+k)} \right)  + H\left({\bf q}^{(1+k)}, {\bf q}^{(0)} \right) = \text{KL}({\bf q}^{(1+k)} || {\bf q}^{(0)} )$.
\end{proof}

\subsection{Switching Equilibria}

Here, we consider the most adversarial possible equilibrium, which we call 'switching equilibria.' The idea is that in a self-resolving prediction market with a rolling window, agents can alternate between extreme reports of $p_1^{(t)} = 0$ and $p_1^{(t-1)}=1$ to maximally modify the previous report while maximally agreeing with the reference agent's report. Doing so maximizes payoff under cross-entropy market scoring rules. We state this formally below:

\begin{theorem}
    In self-resolving prediction markets where reference agents are chosen with a `rolling window' of an even number of timesteps, the strategy profile where every alternate agent reports 1 minus the previous agent's prediction is a perfect Bayesian equilibrium. Additionally, every agent receives an infinite payoff in this equilibrium.
\end{theorem}

\begin{proof}
    For any agent $t$, suppose the previous agent reported $p_1^{(t-1)} = 0$ and the reference agent $r$ is an even number of timesteps away. Further suppose that the reference agent $r$ will report $p_1^{(t+k)} = 1$ regardless of any previous reports.  Then, agent $t$'s payoff under cross-entropy market scoring rules is maximized when reporting ${\bf q}^{(t)} = {\bf q}^{(t+k)} = \text{permute}({\bf q}^{(t-1)})$, and agent $t$ has no incentive to deviate from this strategy. Since this holds for every agent $t$, regardless of agents' beliefs or past agents' reports, the switching profile is a perfect Bayesian equilibrium.

    \begin{equation}\small
    \begin{split}
        \mathbb{E}[S_{CEM}({\bf q}^{(t)}) | \hat{X}^{(t)}]
    &= -H\left({\bf q}^{(t+k)}, {\bf q}^{(t)} \right)  + H\left({\bf q}^{(t+k)}, {\bf q}^{(t-1)} \right) \\
    &= 1\, \log(1) + 0 \,\log(0) - 1 \,\log(0) - 0\, \log(1) \\
    &= 0  + \infty \\
    &= \infty
    \end{split}
    \end{equation}
\end{proof}

An additional challenge highlighted by this result is that the loss to the automated market maker is unbounded since each agent may have a different reference agent in general. In order to get rid of this equilibrium, the mechanism could use a `fixed reference agent', withhold the choice of $k$ from the agents, randomize the reference agent, etc.

\subsection{Permutation Equilibria}

Lastly, we consider permutation equilibria, a mainstay of mechanisms for eliciting information without ground truth. The idea is that all agents could collectively re-label $Y=0$ as $Y=1$ and vice-versa, and the mechanism would function identically. This equilibrium is not a practical concern, since coordinating into such a permutation equilibrium is exceedingly unlikely, but we mention it for completeness.

\begin{theorem}
    In self-resolving prediction markets, the strategy profile where every agent permutes their true belief to report $q_1^{(1)} = p_0^{(t)}$ and $q_0^{(1)} = p_1^{(t)}$, and the belief profile where every agent believes every other agent has/will similarly permute their beliefs is an $\varepsilon-$perfect Bayesian equilibrium.
\end{theorem}

\begin{proof}
    Let us define $\tilde{Y} = \neg Y$. Then, the strategy profile where every agent reports their true belief over $\tilde{Y}$ and the belief profile where every agent believes every other agent reports truthfully (allowing for arbitrary updates if an agent reports an invalid posterior off the equilibrium path) is a $\varepsilon-$perfect Bayesian equilibrium by Theorem \ref{thm:rectru}. Yet, this is equivalent to the strategy and belief profiles with permuted reports and beliefs; hence, the permutation strategy and belief profile is also a $\varepsilon-$perfect Bayesian equilibrium.
\end{proof}

\section{Proofs} \label{app:proofs}

\paragraph{Lemma \ref{lem:post}.} \emph{ Agent $t$'s expectation of agent $r$'s true posterior ${\tilde{p}}_1^{(r)} = \mathbb{P}(Y=1|x^{(r)}, \tilde{x}^{(t)}, x^{(s)})$ is:}
    \begin{equation}\small
       \mathbb{E}[{\tilde{p}}_1^{(r)} | x^{(t)}, x^{(s)}]  =  \mathbb{P}(Y=1|x^{(t)}, x^{(s)})  + \Delta(\Omega_{X_r}, \tilde{x}^{(t)}, {x}^{(t)}, x^{(s)})
    \end{equation}

    \emph{where $\Delta(\Omega_{X_r}, \tilde{x}^{(t)}, {x}^{(t)}, x^{(s)}) = \mu(\Omega_{X_r}, \tilde{x}^{(t)}, x^{(s)}) \cdot \rho(\tilde{x}^{(t)}, x^{(t)}, x^{(s)})$, with }
    \begin{equation}\small
        \begin{split}
            \mu(\Omega_{X_r}, \tilde{x}^{(t)}, x^{(s)}) &= \sum_{x^{(r)} \in \Omega_{X_r}}  \frac{1 }{ \frac{1}{\mathbb{P}(x^{(r)}, Y=1|\tilde{x}^{(t)}, x^{(s)})} + \frac{1}{\mathbb{P}(x^{(r)}, Y=0|\tilde{x}^{(t)}, x^{(s)})}} \\
             \rho(\tilde{x}^{(t)}, x^{(t)}, x^{(s)}) &= \frac{\mathbb{P}(Y=1|\tilde{x}^{(t)}, x^{(s)}) - \mathbb{P}(Y=1|x^{(t)}, x^{(s)})}{\mathbb{P}(Y=0|\tilde{x}^{(t)}, x^{(s)})\cdot  \mathbb{P}(Y=1|\tilde{x}^{(t)}, x^{(s)})} 
        \end{split}
    \end{equation}

\begin{proof}
We denote agent $r$'s potentially misinformed but truthful report as $\tilde{\bf p}^{(r)} = \mathbb{P}(Y|x^{(r)}, \tilde{x}^{(t)}, x^{(s)})$. To be clear, although agent $t$ expects agent $r$ to update as if they they learned $X^{(t)} = \tilde{x}^{(t)}$, agent $t$'s true signal is $X^{(t)} = x^{(t)}$ and so takes the expectation of agent $r$'s information accordingly. Thus, agent $t$'s expectation of agent $r$'s posterior (for $Y=1$) given their signal is $x^{(t)}$ and they report $X^{(t)} = \tilde{x}^{(t)}$ is:

\begin{align}\small
    \mathbb{E}[{\tilde{p}}_1^{(r)} | x^{(t)}, x^{(s)}] &= \mathbb{E}_{Y|x^{(t)}, x^{(s)}}\left(\mathbb{E}[{\tilde{p}}_1^{(r)} | x^{(t)}, x^{(s)}, Y]\right)  \nonumber \\ 
    &= \mathbb{E}_{Y|x^{(t)}, x^{(s)}}\left(\mathbb{E}[{\tilde{p}}_1^{(r)} | Y ]\right) \tag{by Assumption \ref{as:is}} \\ 
    &= \mathbb{P}(Y=0|x^{(t)}, x^{(s)})\cdot \mathbb{E}[{\tilde{p}}_1^{(r)} | Y=0] + \mathbb{P}(Y=1|x^{(t)}, x^{(s)})\cdot \mathbb{E}[{\tilde{p}}_1^{(r)} | Y=1]  \label{eq:expfalsee}
\end{align}

We break down the computation of this expectation into three parts. First we write down the expression for agent $r$'s potentially misinformed posterior $\tilde{\bf p}^{(r)}$, then compute the conditional-on-$Y$ expectations, and subsequently compute Equation \ref{eq:expfalsee}.

\paragraph{Step 1: (Computing ${ \tilde{p}}_1^{(r)}$)} Agent $r$'s posterior ${ \tilde{\bf p}}^{(r)}$ will depend on all the information available to them, including information $x^{(r)}$ that agent $t$ does not have access to:

\begin{equation} \small
    {\tilde{p}}_1^{(r)} = \mathbb{P}(Y=1|x^{(r)} , \tilde{x}^{(t)}, x^{(s)}) 
    \label{eq:wordz_2}
\end{equation}

\paragraph{Step 2: (Computing $\mathbb{E}[{ \tilde{p}}_1^{(r)}|Y]$)} First, we compute the expectation of the misinformed posterior probability ${\tilde{p}}_1^{(r)} = \mathbb{P}(Y=1 | x^{(r)}, \tilde{x}^{(t)}, x^{(s)})$ when in reality $Y=1$:

\begin{equation}\small
\begin{split}
        \mathbb{E}[{\tilde{p}}_1^{(r)} | Y=1] &= \sum_{x^{(r)}} \mathbb{P}(x^{(r)} | Y=1) \cdot \mathbb{P}(Y=1 | x^{(r)}, \tilde{x}^{(t)}, x^{(s)}) \nonumber \\
        &= \sum_{x^{(r)}} \mathbb{P}(x^{(r)} | Y=1) \cdot \left(1- \mathbb{P}(Y=0 | x^{(r)}, \tilde{x}^{(t)}, x^{(s)})\right) \nonumber \\
        &= \sum_{x^{(r)}} \mathbb{P}(x^{(r)} | Y=1)  - \sum_{x^{(r)}} \mathbb{P}(x^{(r)} | Y=1) \cdot \frac{\mathbb{P}(x^{(r)} | Y=0) \mathbb{P}(\tilde{x}^{(t)}, x^{(s)}| Y=0) \mathbb{P}(Y=0)}{\mathbb{P}( x^{(r)}, \tilde{x}^{(t)}, x^{(s)})} \nonumber  \\
        &= 1  - \sum_{x^{(r)}} \frac{\mathbb{P}(x^{(r)} | Y=0) \mathbb{P}(Y=0 | \tilde{x}^{(t)}, x^{(s)}) \cdot \mathbb{P}(x^{(r)} | Y=1)}{\mathbb{P}( x^{(r)}| \tilde{x}^{(t)}, x^{(s)})} \nonumber  \\
        &= 1  - \frac1{\mathbb{P}(Y=1|\tilde{x}^{(t)}, x^{(s)})} \sum_{x^{(r)}}  \frac{\mathbb{P}(x^{(r)}, Y=0 | \tilde{x}^{(t)}, x^{(s)}) \cdot \mathbb{P}(x^{(r)} , Y=1|\tilde{x}^{(t)}, x^{(s)}) }{\mathbb{P}( x^{(r)}, Y=0| \tilde{x}^{(t)}, x^{(s)}) + \mathbb{P}( x^{(r)}, Y=1| \tilde{x}^{(t)}, x^{(s)})} \nonumber  \\
        &= 1  - \frac{\mu(\Omega_{X_r}, \tilde{x}^{(t)}, x^{(s)})}{\mathbb{P}(Y=1|\tilde{x}^{(t)}, x^{(s)})}   
    \end{split}
\end{equation}

where 
\begin{equation}\label{eq:Ltje}\small
\mu(\Omega_{X_r}, \tilde{x}^{(t)}, x^{(s)}) := \sum_{x^{(r)} \in \Omega_{X_r}}  \frac{1 }{ \frac{1}{\mathbb{P}(x^{(r)}, Y=1|\tilde{x}^{(t)}, x^{(s)})} + \frac{1}{\mathbb{P}(x^{(r)}, Y=0|\tilde{x}^{(t)}, x^{(s)})}}
\end{equation}

This is the sum of (half) harmonic means of the joint probability of the agent $r$'s private information $x^{(r)}$ under $Y=1$ and $Y=0$ given information $x^{(s)}$ and $X^{(t)} = \tilde{x}^{(t)}$, summed over all of agent $r$'s possible private signals $x^{(r)}$.
Similarly, we can compute the expectation of the misinformed posterior probability ${\tilde{p}}_1^{(r)} = \mathbb{P}(Y=1 |x^{(r)},  \tilde{x}^{(t)}, x^{(s)})$ when in reality $Y=0$:

\begin{align}\small
        \mathbb{E}[{\tilde{p}}_1^{(r)} | Y=0] &= \sum_{x^{(r)}} \mathbb{P}(x^{(r)} | Y=0) \cdot \mathbb{P}(Y=1 | x^{(r)}, \tilde{x}^{(t)}, x^{(s)}) \nonumber \\
        &=  \sum_{x^{(r)}} \mathbb{P}(x^{(r)} | Y=0) \cdot \frac{\mathbb{P}(x^{(r)} | Y=1) \mathbb{P}(\tilde{x}^{(t)}, x^{(s)}| Y=1) \mathbb{P}(Y=1)}{\mathbb{P}( x^{(r)}, \tilde{x}^{(t)}, x^{(s)})} \nonumber  \\
        &= \sum_{x^{(r)}} \frac{\mathbb{P}(x^{(r)} | Y=0)  \cdot \mathbb{P}(x^{(r)} | Y=1)\mathbb{P}(Y=1 | \tilde{x}^{(t)}, x^{(s)})}{\mathbb{P}( x^{(r)}| \tilde{x}^{(t)}, x^{(s)})} \nonumber  \\
        &= \frac1{\mathbb{P}(Y=0|\tilde{x}^{(t)}, x^{(s)})} \sum_{x^{(r)}}  \frac{\mathbb{P}(x^{(r)}, Y=0 | \tilde{x}^{(t)}, x^{(s)}) \cdot \mathbb{P}(x^{(r)} , Y=1|\tilde{x}^{(t)}, x^{(s)}) }{\mathbb{P}( x^{(r)}, Y=0| \tilde{x}^{(t)}, x^{(s)}) + \mathbb{P}( x^{(r)}, Y=1| \tilde{x}^{(t)}, x^{(s)})} \nonumber  \\
        &=  \frac{\mu(\Omega_{X_r}, \tilde{x}^{(t)}, x^{(s)})}{\mathbb{P}(Y=0|\tilde{x}^{(t)}, x^{(s)})} 
\end{align}

\paragraph{Step 3: (Computing $\mathbb{E}[{\tilde{p}}_1^{(r)} | x^{(t)}, x^{(s)}])$}  Finally, we can compute agent $t$'s expectation of agent $r$'s posterior when they report the posterior associated with $\tilde{x}^{(t)}$ using Equation \ref{eq:expfalsee}:

\begin{equation}\small
    \begin{split}
    \mathbb{E}[{\tilde{p}}_1^{(r)} | x^{(t)}, x^{(s)}] 
    &= \mathbb{P}(Y=0|x^{(t)}, x^{(s)})\cdot \mathbb{E}[{\tilde{p}}_1^{(r)} | Y=0]  + \mathbb{P}(Y=1|x^{(t)}, x^{(s)})\cdot \mathbb{E}[{\tilde{p}}_1^{(r)} | Y=1] \\
    &= \mathbb{P}(Y=0|x^{(t)}, x^{(s)}) \cdot \frac{\mu(\Omega_{X_r}, \tilde{x}^{(t)}, x^{(s)})}{\mathbb{P}(Y=0|\tilde{x}^{(t)}, x^{(s)})}  + \mathbb{P}(Y=1|x^{(t)}, x^{(s)}) \cdot \left(1 - \frac{\mu(\Omega_{X_r}, \tilde{x}^{(t)}, x^{(s)})}{\mathbb{P}(Y=1|\tilde{x}^{(t)}, x^{(s)})}  \right) \\ 
    &= \mathbb{P}(Y=1|x^{(t)}, x^{(s)})  + \mu(\Omega_{X_r}, \tilde{x}^{(t)}, x^{(s)})\left( \frac{\mathbb{P}(Y=0|x^{(t)}, x^{(s)})}{\mathbb{P}(Y=0|\tilde{x}^{(t)}, x^{(s)})}  - \frac{\mathbb{P}(Y=1|x^{(t)}, x^{(s)})}{\mathbb{P}(Y=1|\tilde{x}^{(t)}, x^{(s)})} \right) \\ 
    &= \mathbb{P}(Y=1|x^{(t)}, x^{(s)})  + \mu(\Omega_{X_r}, \tilde{x}^{(t)}, x^{(s)})\left( \frac{\mathbb{P}(Y=1|\tilde{x}^{(t)}, x^{(s)}) - \mathbb{P}(Y=1|x^{(t)}, x^{(s)})}{\mathbb{P}(Y=0|\tilde{x}^{(t)}, x^{(s)}) \cdot \mathbb{P}(Y=1|\tilde{x}^{(t)}, x^{(s)})}  \right) \\ 
    &= \mathbb{P}(Y=1|x^{(t)}, x^{(s)})  + \mu(\Omega_{X_r}, \tilde{x}^{(t)}, x^{(s)}) \cdot \rho(\tilde{x}^{(t)}, x^{(t)}, x^{(s)}) 
\end{split}
\end{equation}

where we defined:
\begin{equation} \small
    \rho(\tilde{x}^{(t)}, x^{(t)}, x^{(s)}) := \frac{\mathbb{P}(Y=1|\tilde{x}^{(t)}, x^{(s)}) - \mathbb{P}(Y=1|x^{(t)}, x^{(s)})}{\mathbb{P}(Y=0|\tilde{x}^{(t)}, x^{(s)})\cdot  \mathbb{P}(Y=1|\tilde{x}^{(t)}, x^{(s)})}  
\end{equation}

\end{proof}


\paragraph{Theorem \ref{eq:imptheorem}. }
    \emph{Suppose a reference agent $r$ can observe $k$ private signals $x^{(1)}, \ldots,  x^{(k)}$ that are informational substitutes, where $\Omega_{X^{(r)}} = \Omega_{X^{(1)}} \times \ldots \times \Omega_{X^{(k)}}$ and $x^{(j)} \in \Omega_{X^{(j)}}$, and agent $t$ cannot observe these signals.  Then, agent $t$'s adjustment $\Delta$ of their posterior to predict their reference agent's  true posterior is upper bounded in magnitude as $|\Delta(\Omega_{X_r}, \tilde{x}^{(t)}, x^{(t)}, x^{(s)}) |\leq \frac14 \left({{\frac{1-\eta}{\eta}}} - {{\frac{\eta}{1-\eta}}} \right) \left( 1-\delta \right)^k $.
    Consequently, we have an upper bound $|\Delta | \leq \varepsilon'$ for any $k$ where}
    \begin{equation}\small
    k  \geq \frac1{-\log \left( 1-\delta \right)}\log \left(\frac{1}{4\varepsilon'}\left({{\frac{1-\eta}{\eta}}} - {{\frac{\eta}{1-\eta}}} \right)\right)
    \end{equation}

\begin{proof}
    $\mu(\Omega_{X_r}, \tilde{x}^{(t)}, x^{(s)})$  (defined in Equation \ref{eq:Ltje})is non-negative (as a harmonic mean of probabilities) and we can upper bound it using the property that the harmonic mean is always less than the geometric mean:

\begin{equation}\small
\begin{split}
    \mu(\Omega_{X_r}, \tilde{x}^{(t)}, x^{(s)}) &= \sum_{x^{(r)} \in \Omega_{X_r}}  \frac{1 }{ \frac{1}{\mathbb{P}(x^{(r)}, Y=1|\tilde{x}^{(t)}, x^{(s)})} + \frac{1}{\mathbb{P}(x^{(r)}, Y=0|\tilde{x}^{(t)}, x^{(s)})}} \\
    &\leq \sum_{x^{(r)} \in \Omega_{X_r}} \frac12 \left(\mathbb{P}(x^{(r)}, Y=0|\tilde{x}^{(t)}, x^{(s)}) \cdot \mathbb{P}(x^{(r)}, Y=1|\tilde{x}^{(t)}, x^{(s)}) \right)^\frac12 \\
    &= \sum_{x^{(r)} \in \Omega_{X_r}} \frac12 \left[ \mathbb{P}(x^{(r)}| Y=0)\mathbb{P}(Y=0|\tilde{x}^{(t)}, x^{(s)})\right]^\frac12\left[ \mathbb{P}(x^{(r)}| Y=1)\mathbb{P}(Y=1|\tilde{x}^{(t)}, x^{(s)})\right]^\frac12 
 \\
 &= \frac12 \left(\mathbb{P}(Y=0|\tilde{x}^{(t)}, x^{(s)}) \cdot \mathbb{P}(Y=1|\tilde{x}^{(t)}, x^{(s)})\right)^\frac12 \sum_{x^{(r)} \in \Omega_{X_r}} \left(\prod_{i=1}^k \left( \mathbb{P}(x^{(r_i)}| Y=0) \cdot  \mathbb{P}(x^{(r_i)}| Y=1)\right)^\frac12 \right)
 \\
 &= \frac12 \left(\mathbb{P}(Y=0|\tilde{x}^{(t)}, x^{(s)}) \cdot \mathbb{P}(Y=1|\tilde{x}^{(t)}, x^{(s)})\right)^\frac12 \prod_{i=1}^k \left(\sum_{j=1}^N \left( \mathbb{P}(x^{(r_i)}_j| Y=0) \cdot  \mathbb{P}(x^{(r_i)}_j| Y=1)\right)^\frac12 \right)
 \\
 &= \frac12 \left(\mathbb{P}(Y=0|\tilde{x}^{(t)}, x^{(s)}) \cdot \mathbb{P}(Y=1|\tilde{x}^{(t)}, x^{(s)})\right)^\frac12 \prod_{i=1}^k \left( 1-\delta^{(r_i)} \right)
 \\
 &\leq \frac12 \left(\mathbb{P}(Y=0|\tilde{x}^{(t)}, x^{(s)}) \cdot \mathbb{P}(Y=1|\tilde{x}^{(t)}, x^{(s)})\right)^\frac12  \left( 1-\delta \right)^k 
\end{split}\label{eq:Lup}
\end{equation}

where $1-\delta^{(r_i)}$ is the Bhattacharyya coefficient for signal $x^{(r_i)}$ (Definition \ref{def:bhat}). Thus, we have:

\begin{equation}\small
    |\Delta| \leq \frac12 \left( 1-\delta \right)^k \frac{\mathbb{P}(Y=1|\tilde{x}^{(t)}, x^{(s)}) - \mathbb{P}(Y=1|x^{(t)}, x^{(s)})}{\left(\mathbb{P}(Y=0|\tilde{x}^{(t)}, x^{(s)})\cdot  \mathbb{P}(Y=1|\tilde{x}^{(t)}, x^{(s)})\right)^\frac12}
\end{equation}

Now, observe that we can write the posterior as:

\begin{equation} \small
\begin{split}
    \mathbb{P}(Y=1|\tilde{x}^{(t)}, x^{(s)}) 
    &= \frac{\mathbb{P}(\tilde{x}^{(t)}|Y=1) \mathbb{P}(Y=1|x^{(s)})\mathbb{P}(x^{(s)})}{\mathbb{P}(\tilde{x}^{(t)}|Y=0) \mathbb{P}(Y=0|x^{(s)})\mathbb{P}(x^{(s)}) + \mathbb{P}(\tilde{x}^{(t)}|Y=1) \mathbb{P}(Y=1|x^{(s)})\mathbb{P}(x^{(s)})}\\
    &= \frac{\mathbb{P}(\tilde{x}^{(t)}|Y=1) \mathbb{P}(Y=1|x^{(s)})}{\mathbb{P}(\tilde{x}^{(t)}|x^{(s)})} 
\end{split}
\end{equation}

Then, we can manipulate the term in the upper bound as:

\begin{equation}\tiny
\begin{split}
\frac{\mathbb{P}(Y=1|\tilde{x}^{(t)}, x^{(s)}) - \mathbb{P}(Y=1|x^{(t)}, x^{(s)})}{\left(\mathbb{P}(Y=0|\tilde{x}^{(t)}, x^{(s)})\cdot  \mathbb{P}(Y=1|\tilde{x}^{(t)}, x^{(s)})\right)^\frac12} &= \sqrt{\frac{\mathbb{P}(Y=1|\tilde{x}^{(t)}, x^{(s)})}{\mathbb{P}(Y=0|\tilde{x}^{(t)}, x^{(s)})}} - \frac{\mathbb{P}(Y=1|{x}^{(t)}, x^{(s)})}{\sqrt{\mathbb{P}(Y=1|\tilde{x}^{(t)}, x^{(s)})\mathbb{P}(Y=0|\tilde{x}^{(t)}, x^{(s)})}} \\
&= \sqrt{\frac{\mathbb{P}(\tilde{x}^{(t)} | Y=1)\mathbb{P}(Y=1| x^{(s)})}{\mathbb{P}(\tilde{x}^{(t)} | Y=0)\mathbb{P}(Y=0| x^{(s)})}} \\
&\phantom{=} -\frac{\mathbb{P}(\tilde{x}^{(t)}|x^{(s)})}{\mathbb{P}({x}^{(t)}|x^{(s)})}\frac{\mathbb{P}({x}^{(t)}|Y=1) \sqrt{\mathbb{P}(Y=1|x^{(s)})}}{\sqrt{\mathbb{P}(\tilde{x}^{(t)}|Y=1) \mathbb{P}(\tilde{x}^{(t)}|Y=0) \mathbb{P}(Y=0|x^{(s)})}} \\
&= \frac{\sqrt{\mathbb{P}(Y=1| x^{(s)})}}{\mathbb{P}({x}^{(t)}|x^{(s)})\sqrt{\mathbb{P}(\tilde{x}^{(t)} | Y=0)\mathbb{P}(Y=0|x^{(s)})}} \\
&\phantom{=} \cdot \left(\frac{\mathbb{P}({x}^{(t)}|x^{(s)}) {\mathbb{P}(\tilde{x}^{(t)} | Y=1)} - \mathbb{P}(\tilde{x}^{(t)}|x^{(s)})\mathbb{P}({x}^{(t)}|Y=1) }{\sqrt{\mathbb{P}(\tilde{x}^{(t)}|Y=1)  }} \right) \\
&= \frac{\sqrt{\mathbb{P}(Y=1| x^{(s)})\mathbb{P}(Y=0|x^{(s)})}}{\mathbb{P}({x}^{(t)}|x^{(s)})\sqrt{\mathbb{P}(\tilde{x}^{(t)} | Y=0)}} \\
&\phantom{=} \cdot \left(\frac{\mathbb{P}(\tilde{x}^{(t)} | Y=1)\mathbb{P}({x}^{(t)} | Y=0) - \mathbb{P}(\tilde{x}^{(t)} | Y=0)\mathbb{P}({x}^{(t)} | Y=1) }{\sqrt{\mathbb{P}(\tilde{x}^{(t)}|Y=1)  }} \right) \\
&= \frac{\sqrt{\mathbb{P}(Y=1| x^{(s)})\mathbb{P}(Y=0|x^{(s)})}}{\mathbb{P}({x}^{(t)}|x^{(s)})} \\
&\phantom{=} \cdot \left(\frac{\mathbb{P}(\tilde{x}^{(t)} | Y=1)\mathbb{P}({x}^{(t)} | Y=0)}{\sqrt{\mathbb{P}(\tilde{x}^{(t)}|Y=1)\mathbb{P}(\tilde{x}^{(t)} | Y=0)}} - \frac{\mathbb{P}(\tilde{x}^{(t)} | Y=0)\mathbb{P}({x}^{(t)} | Y=1) }{\sqrt{\mathbb{P}(\tilde{x}^{(t)}|Y=1)\mathbb{P}(\tilde{x}^{(t)} | Y=0)}} \right) \\
&= \frac{\sqrt{\mathbb{P}(Y=1| x^{(s)})\mathbb{P}(Y=0|x^{(s)})}}{\mathbb{P}({x}^{(t)}|x^{(s)})} \\
&\phantom{=} \cdot \left({\mathbb{P}({x}^{(t)} | Y=0)}{\sqrt{\frac{\mathbb{P}(\tilde{x}^{(t)}|Y=1)}{\mathbb{P}(\tilde{x}^{(t)} | Y=0)}}} - \mathbb{P}({x}^{(t)} | Y=1) {\sqrt{\frac{\mathbb{P}(\tilde{x}^{(t)} | Y=0)}{\mathbb{P}(\tilde{x}^{(t)}|Y=1)}}} \right) \\
&\leq \frac{\sqrt{\mathbb{P}(Y=1| x^{(s)})\mathbb{P}(Y=0|x^{(s)})}}{\mathbb{P}({x}^{(t)}|Y=0) \mathbb{P}(Y=0|x^{(s)}) + \mathbb{P}({x}^{(t)}|Y=1) \mathbb{P}(Y=1|x^{(s)})} \\
&\phantom{=} \cdot \left({\mathbb{P}({x}^{(t)} | Y=0)}{\sqrt{\frac{1-\eta}{\eta}}} - \mathbb{P}({x}^{(t)} | Y=1) {\sqrt{\frac{\eta}{1-\eta}}} \right) \\
&\leq \frac{\sqrt{\frac{\mathbb{P}({x}^{(t)}|Y=0)\mathbb{P}({x}^{(t)}|Y=1)}{(\mathbb{P}({x}^{(t)}|Y=0) + \mathbb{P}({x}^{(t)}|Y=1))^2}}}{ \frac{\mathbb{P}({x}^{(t)}|Y=0)\mathbb{P}({x}^{(t)}|Y=1)}{\mathbb{P}({x}^{(t)}|Y=0) + \mathbb{P}({x}^{(t)}|Y=1)} +  \frac{\mathbb{P}({x}^{(t)}|Y=1)\mathbb{P}({x}^{(t)}|Y=0)}{\mathbb{P}({x}^{(t)}|Y=0) + \mathbb{P}({x}^{(t)}|Y=1)}} \\
&\phantom{=} \cdot \left({\mathbb{P}({x}^{(t)} | Y=0)}{\sqrt{\frac{1-\eta}{\eta}}} - \mathbb{P}({x}^{(t)} | Y=1) {\sqrt{\frac{\eta}{1-\eta}}} \right) \\
&= \frac{1}{ 2 \cdot \sqrt{\mathbb{P}({x}^{(t)}|Y=0)\mathbb{P}({x}^{(t)}|Y=1)} }  \\
&\phantom{=} \cdot \left({\mathbb{P}({x}^{(t)} | Y=0)}{\sqrt{\frac{1-\eta}{\eta}}} - \mathbb{P}({x}^{(t)} | Y=1) {\sqrt{\frac{\eta}{1-\eta}}} \right) \\
&= \frac{1}{ 2} \left(\sqrt{\frac{\mathbb{P}({x}^{(t)}|Y=0)}{\mathbb{P}({x}^{(t)}|Y=1)}}{\sqrt{\frac{1-\eta}{\eta}}} - \sqrt{\frac{\mathbb{P}({x}^{(t)}|Y=1)}{\mathbb{P}({x}^{(t)}|Y=0)}} {\sqrt{\frac{\eta}{1-\eta}}} \right) \\
&\leq \frac{1}{ 2} \left({\sqrt{\frac{1-\eta}{\eta}}\sqrt{\frac{1-\eta}{\eta}}} - \sqrt{\frac{\eta}{1-\eta}}{\sqrt{\frac{\eta}{1-\eta}}} \right) \\
&= \frac{1}{ 2} \left({{\frac{1-\eta}{\eta}}} - {{\frac{\eta}{1-\eta}}} \right) 
\end{split}
\end{equation}

By a similar procedure, we can establish that $- \frac{\mathbb{P}(Y=1|\tilde{x}^{(t)}, x^{(s)}) - \mathbb{P}(Y=1|x^{(t)}, x^{(s)})}{\left(\mathbb{P}(Y=0|\tilde{x}^{(t)}, x^{(s)})\cdot  \mathbb{P}(Y=1|\tilde{x}^{(t)}, x^{(s)})\right)^\frac12} \leq \frac{1}{ 2} \left({{\frac{1-\eta}{\eta}}} - {{\frac{\eta}{1-\eta}}} \right)$. Thus, we have that:

\begin{equation}\small
\begin{split}
    \left|\mu(\Omega_{X_r}, \tilde{x}^{(t)}, x^{(s)}) \cdot \rho(\tilde{x}^{(t)}, {x}^{(t)}, x^{(s)}) \right| &\leq \frac14 \left({{\frac{1-\eta}{\eta}}} - {{\frac{\eta}{1-\eta}}} \right) \left( 1-\delta \right)^k
\end{split} \label{eq:bound1}
\end{equation}

Thus, we can upper bound Equation \ref{eq:bound1} with any desired upper bound $\varepsilon$ and invert to find $k$ such that $|\mu(\Omega_{X_r}, \tilde{x}^{(t)}, x^{(s)})\cdot \rho(\tilde{x}^{(t)}, {x}^{(t)}, x^{(s)})| \leq \varepsilon$:

\begin{equation} \small
      k  \geq \frac1{-\log \left( 1-\delta \right)}\log \left(\frac{1}{4\varepsilon}\left({{\frac{1-\eta}{\eta}}} - {{\frac{\eta}{1-\eta}}} \right)\right)
\end{equation}

\end{proof}


\paragraph{Theorem \ref{thm:exact}.}
    \emph{Suppose agent $t$ is paid based on a cross-entropy scoring rule with respect to agent $r$ who reports truthfully.  Then, agent $t$ strictly maximizes their expected payoff by reporting truthfully if the reference agent has at least $k$ signals where}
    \begin{equation}
        k  > \frac{1}{-\log \left( 1-\delta \right)}\log\left(\frac{\left|\log \left(\frac{1-\eta}{\eta}\right)\right|\left({{\frac{1-\eta}{\eta}}} - {{\frac{\eta}{1-\eta}}} \right)}{8 \left(\tau\eta(1-\eta)\right)^2} \right)
    \end{equation}

\begin{proof}
We proceed as follows:

\textbf{Step 1: Condition for strict truthfulness} In order for truthful reporting to pay more than any strategic misreport, we consider the difference in expected scores when reporting the truth ${\bf p}^{(t)}$ relative to reporting some other potentially better report ${\bf q}^{(t)}$ (that is consistent with signal $\tilde{x}^{(t)}$:
\begin{equation}
\begin{split}
     \mathbb{E}[S_{CE}({\bf q}^{(t)}) | \hat{X}^{(t)}] - \mathbb{E}[S_{CE}({\bf p}^{(t)}) | \hat{X}^{(t)}]
     &= -H\left(\mathbb{E} \left[ \tilde{\bf p}^{(r)} \left.\right |\hat{X}^{(t)} \right], {\bf q}^{(t)} \right) + H\left( {\bf p}^{(t)} \right) \\
     &= -H\left({\bf p}^{(t)} + \Delta, {\bf q}^{(t)} \right) + H\left( {\bf p}^{(t)} \right) \\
     &= ({p_1}^{(t)} + \Delta) \log \left( {{q_1}^{(t)}} \right) + (1-{p_1}^{(t)} - \Delta) \log \left( {1-{q_1}^{(t)}} \right) \\
     &\phantom{=} - {p_1}^{(t)} \log \left( {{p_1}^{(t)} } \right) - (1-{p_1}^{(t)} ) \log \left( {1-{p_1}^{(t)} } \right) \\
     &= ({p_1}^{(t)} + \Delta) \log \left( {{q_1}^{(t)}} \right) + (1-{p_1}^{(t)} - \Delta) \log \left( {1-{q_1}^{(t)}} \right) \\
     &\phantom{=} - {p_1}^{(t)} \log \left( {{p_1}^{(t)} } \right) - (1-{p_1}^{(t)} ) \log \left( {1-{p_1}^{(t)} } \right) \\
     &= \Delta \log \left(\frac{q_1^{(t)}}{1-q_1^{(t)}}\right) - \text{KL}({\bf p}^{(t)} \| {\bf q}^{(t)})
\end{split}
\end{equation}
Thus, truthful reporting yields higher utility over reporting ${\bf q}^{(t)}$ when:
\begin{equation}
    \Delta \log \left(\frac{q_1^{(t)}}{1-q_1^{(t)}}\right)  - \text{KL}({\bf p}^{(t)} \| {\bf q}^{(t)}) < 0
\end{equation}

The KL term is always positive. We see that when $\frac12 < q_1^{(t)} < p_1^{(t)}$, we have $\log \left(\frac{q_1^{(t)}}{1-q_1^{(t)}}\right) > 0$ and $\Delta < 0$ (since $\rho < 0$), so the condition is satisfied. Similarly, when $\frac12 > q_1^{(t)} > p_1^{(t)} $, we have $\log \left(\frac{q_1^{(t)}}{1-q_1^{(t)}}\right) < 0$ and $\Delta > 0$, so again the condition is satisfied. Thus, the only way agent $t$ can extract greater reward is by reporting more \emph{extreme} beliefs: when $\frac12 \leq p_1^{(t)} < q_1^{(t)}$ or $\frac12 \geq p_1^{(t)} > q_1^{(t)}$, leveraging the idea that approximate agreement at the extremes is higher payoff than exact agreement in the middle given the convexity of cross-entropy.

Thus, truthful reporting yields higher utility over reporting ${\bf q}^{(t)}$ when:
\begin{equation}
    |\Delta| < \frac{\text{KL}({\bf p}^{(t)} \| {\bf q}^{(t)})}{\left|\log \left(\frac{q_1^{(t)}}{1-q_1^{(t)}}\right)\right|}
\end{equation}

\textbf{Step 2: Lower bound for KL Divergence} Next, observe that by Pinsker's inequality, we have:

\begin{equation}
    \begin{split}
        \text{KL}({\bf p}^{(t)} \| {\bf q}^{(t)}) \geq  2 \left|p_1^{(t)} - q_1^{(t)}\right|^2
    \end{split}
\end{equation}

Writing $\lambda_p = \log\left(\frac{p_1^{(t)}}{1-p_1^{(t)}}\right)$, $\lambda_q = \log\left(\frac{q_1^{(t)}}{1-q_1^{(t)}}\right)$, and $\phi(\lambda) = \log(1 + \exp\left(\lambda\right))$, we have that $p_1^{(t)} = \phi'(\lambda_p)$ and $q_1^{(t)} = \phi'(\lambda_q)$. Thus, we can re-write the KL bound as:

\begin{equation}
    \begin{split}
        \text{KL}({\bf p}^{(t)} \| {\bf q}^{(t)}) \geq  2 \left|\phi'(\lambda_p) - \phi'(\lambda_q)\right|^2
    \end{split}
\end{equation}

By the mean value theorem, we have that $\left|\phi'(\lambda_p) - \phi'(\lambda_q)\right| = \phi''(\tilde{\lambda})\left|\lambda_p - \lambda_q\right|$ for some $\tilde{\lambda}$ in the interval between $\lambda_p$ and $\lambda_q$, allowing us to write:
\begin{equation}
    \begin{split}
        \text{KL}({\bf p}^{(t)} \| {\bf q}^{(t)}) \geq  2 \left(\phi''(\tilde{\lambda})\right)^2\left|\lambda_p - \lambda_q\right|^2
    \end{split}
\end{equation}

Next, we can pull out the prior information. Let the log odds of the prior be $\lambda_y = \log\left(\frac{y_1}{y_0}\right)$, the log likelihood of the shared information be $\lambda_s = \log\left(\frac{\mathbb{P}(x^{(s)}|Y=1)}{\mathbb{P}(x^{(s)}|Y=0)}\right)$, and the log-likelihood of the reported signal be $\lambda_{\tilde{x}} = \log\left(\frac{\mathbb{P}(\tilde{x}^{(t)}|Y=1)}{\mathbb{P}(\tilde{x}^{(t)}|Y=0)}\right)$. Then, if ${\bf x}^{(t)}$ is the signal that induced belief ${\bf p}^{(t)}$ and ${\bf {\tilde{x}}}^{(t)}$ the signal that would induce belief ${\bf q}^{(t)}$, by Assumption \ref{as:is} we have that $\lambda_p = \lambda_y + \lambda_s + \lambda_x$ and $\lambda_q = \lambda_y + \lambda_s + \lambda_{\tilde{x}}$. Then, we can re-write the KL as:
\begin{equation}
        \text{KL}({\bf p}^{(t)} \| {\bf q}^{(t)}) \geq  2 \left(\phi''(\tilde{\lambda})\right)^2\left|\lambda_x - \lambda_{\tilde{x}}\right|^2
\end{equation}
Then, from $\tau$-granularity where $\tau = \min_{x, \tilde{x}} |\lambda_x - \lambda_{\tilde{x}}|$ and $\phi''(\lambda) = \sigma(\tilde{\lambda})(1-\sigma(\tilde{\lambda}) \geq \eta(1-\eta)$, allowing us to write:
\begin{equation}
        \text{KL}({\bf p}^{(t)} \| {\bf q}^{(t)}) \geq  2 \left(\eta(1-\eta)\tau\right)^2
\end{equation}

\textbf{Step 3: Incorporating the KL bound into the condition} Re-visiting the inequality for $\Delta$, and observing that ${\left|\log \left(\frac{q_1^{(t)}}{1-q_1^{(t)}}\right)\right|} = |\lambda_q| = |\lambda_y + \lambda_s + \lambda_{\tilde{x}}| \leq {\left|\lambda_y + 2\log \left(\frac{1-\eta}{\eta}\right)\right|}$, we can write:

\begin{equation}
    \frac{\text{KL}({\bf p}^{(t)} \| {\bf q}^{(t)})}{\left|\log \left(\frac{q_1^{(t)}}{1-q_1^{(t)}}\right)\right|} \geq \frac{2 \left(\eta(1-\eta)\tau\right)^2}{\left|\lambda_y + 2\log \left(\frac{1-\eta}{\eta}\right)\right|}
\end{equation}

This gives us an upper bound for $\Delta$ that will ensure strict truthfulness:

\begin{equation}
    \begin{split}
        |\Delta| < \frac{2 \left(\eta(1-\eta)\tau\right)^2}{\left|\lambda_y + 2\log \left(\frac{1-\eta}{\eta}\right)\right|}
    \end{split}
\end{equation}

Using the upper bound on $\Delta$ from Theorem \ref{eq:imptheorem}, we can identify the needed $k$:
\begin{equation}
    \begin{split}
        \frac14 \left({{\frac{1-\eta}{\eta}}} - {{\frac{\eta}{1-\eta}}} \right) \left( 1-\delta \right)^k &< \frac{2 \left(\eta(1-\eta)\tau\right)^2}{\left|\lambda_y + 2\log \left(\frac{1-\eta}{\eta}\right)\right|} \\
        \Rightarrow \left( 1-\delta \right)^k &< \frac{8 \left(\eta(1-\eta)\tau\right)^2}{\left|\lambda_y + 2\log \left(\frac{1-\eta}{\eta}\right)\right|\left({{\frac{1-\eta}{\eta}}} - {{\frac{\eta}{1-\eta}}} \right)} \\
        \Rightarrow k \log \left( 1-\delta \right) &< \log\left(\frac{8 \left(\eta(1-\eta)\tau\right)^2}{\left|\lambda_y + 2\log \left(\frac{1-\eta}{\eta}\right)\right|\left({{\frac{1-\eta}{\eta}}} - {{\frac{\eta}{1-\eta}}} \right)} \right) \\
        \Rightarrow k  &> \frac{1}{-\log \left( 1-\delta \right)}\log\left(\frac{\left|\log \left(\frac{y_1(1-\eta)^2}{y_0\eta^2}\right)\right|\left({{\frac{1-\eta}{\eta}}} - {{\frac{\eta}{1-\eta}}} \right)}{8 \left(\eta(1-\eta)\tau\right)^2} \right)
    \end{split}
\end{equation}
\end{proof}


\paragraph{Theorem \ref{thm:bound}}
    \emph{If agent $t$ is paid based on a cross-entropy scoring rule with respect to agent $r$ who reports truthfully, the difference in expected payoff when deviating from truthful report ${\bf p}^{(t)}$ associated with signal $x^{(t)}$ to misreport ${\bf q}^{(t)}$ is:}
    \begin{equation}\small 
       \mathbb{E}[S_{CE}({\bf q}^{(t)}) | x^{(t)}, x^{(s)}] - \mathbb{E}[S_{CE}({\bf p}^{(t)}) | x^{(t)}, x^{(s)}] \leq \mathcal{D}_{\eta}(\Delta, {\bf y})
    \end{equation}
    \emph{where ${\bf y}$ is the prior and}
    \begin{equation} \small
        \mathcal{D}_{\eta}(\Delta, {\bf y}) = \Delta \cdot \log \left(\frac{ ({1-\eta})  y_1 +  \Delta \cdot \left(\eta y_0 + ({1-\eta}) {y_1}\right)}{\eta y_0-  \Delta\cdot \left(\eta y_0 + ({1-\eta}){y_1}\right)} \right)  
    \end{equation}
    \emph{where $\Delta$ is as defined in Lemma \ref{lem:post}.}

\begin{proof} By Corollary \ref{cor:mart}, we have that the expected reward for truth-telling is the negative entropy of the true posterior: $\mathbb{E}[S_{CE}({\bf p}^{(t)}) | \hat{X}^{(t)}] = H\left(\mathbb{E} \left[ {\bf p}^{(r)} \left.\right |\hat{X}^{(t)} \right], {\bf p}^{(t)} \right) = -H\left({\bf p}^{(t)} \right)$. Additionally, we can write: 

\begin{align}\small
    \mathbb{E}[S_{CE}({\bf q}^{(t)}) | \hat{X}^{(t)}]
    &= -H\left(\mathbb{E} \left[ \tilde{\bf p}^{(r)} \left.\right |\hat{X}^{(t)} \right], {\bf q}^{(t)} \right)  \nonumber \\
    &\leq -H\left(\mathbb{E} \left[ \tilde{\bf p}^{(r)} \left.\right |\hat{X}^{(t)} \right] \right)  \nonumber \\
    &\leq -H\left({\bf p}^{(t)} \right) + \nabla \left(- H\left(\mathbb{E} \left[ \tilde{\bf p}^{(r)} \left.\right |\hat{X}^{(t)} \right] \right)\right)^T\left( \mathbb{E} \left[ \tilde{\bf p}^{(r)} \left.\right |\hat{X}^{(t)} \right] - {\bf p}^{(t)} \right)  \tag{by concavity of entropy)} \\
    &= -H\left({\bf p}^{(t)} \right) +  \Delta \cdot \log \left(\frac{{p}_1^{(t)} +  \Delta}{1 - {p}_1^{(t)} -  \Delta} \right) \label{eq:cebound}
\end{align}

We can further simplify the term $\Delta \cdot \log \left(\frac{{p}_1^{(t)} +  \Delta}{1 - {p}_1^{(t)} -  \Delta} \right)$ as follows:

\begin{equation}\small
\begin{split}
    \Delta \cdot \log \left(\frac{{p}_1^{(t)} +  \Delta}{1 - {p}_1^{(t)} -  \Delta} \right) &= \Delta \cdot \log \left(\frac{\mathbb{P}(Y=1|x^{(t)}, x^{(s)}) +  \Delta}{\mathbb{P}(Y=0|x^{(t)}, x^{(s)}) -  \Delta} \right)\\
     &= \Delta \cdot \log \left(\frac{\frac{\mathbb{P}(x^{(t)} |Y=1) \mathbb{P}(Y=1|x^{(s)})}{\mathbb{P}(x^{(t)}|x^{(s)})} +  \Delta}{\frac{\mathbb{P}(x^{(t)} |Y=0) \mathbb{P}(Y=0|x^{(s)})}{\mathbb{P}(x^{(t)}|x^{(s)})} -  \Delta} \right) \\
     &= \Delta \cdot \log \left(\frac{ \frac{\mathbb{P}(x^{(t)} |Y=1) \mathbb{P}(Y=1|x^{(s)})}{\mathbb{P}(x^{(t)} |Y=0) \mathbb{P}(Y=0|x^{(s)})} +  \Delta \cdot \frac{\mathbb{P}(x^{(t)}|x^{(s)})}{\mathbb{P}(x^{(t)} |Y=0) \mathbb{P}(Y=0|x^{(s)})}}{1 -  \Delta\cdot \frac{\mathbb{P}(x^{(t)}|x^{(s)})}{\mathbb{P}(x^{(t)} |Y=0) \mathbb{P}(Y=0|x^{(s)})}} \right) \\
     &= \Delta \cdot \log \left(\frac{ \frac{\mathbb{P}(x^{(t)} |Y=1) \mathbb{P}(Y=1|x^{(s)})}{\mathbb{P}(x^{(t)} |Y=0) \mathbb{P}(Y=0|x^{(s)})} +  \Delta \cdot \left(1 + \frac{\mathbb{P}(x^{(t)} |Y=1) \mathbb{P}(Y=1|x^{(s)})}{\mathbb{P}(x^{(t)} |Y=0) \mathbb{P}(Y=0|x^{(s)})}\right)}{1 -  \Delta\cdot \left(1 + \frac{\mathbb{P}(x^{(t)} |Y=1) \mathbb{P}(Y=1|x^{(s)})}{\mathbb{P}(x^{(t)} |Y=0) \mathbb{P}(Y=0|x^{(s)})}\right)} \right) \\
     &\leq \Delta \cdot \log \left(\frac{ \frac{1-\eta}{\eta} \cdot \frac{y_1}{ y_0} +  \Delta \cdot \left(1 + \frac{1-\eta}{\eta} \cdot \frac{y_1}{ y_0}\right)}{1 -  \Delta\cdot \left(1 + \frac{1-\eta}{\eta} \cdot \frac{y_1}{ y_0}\right)} \right) \\
     &= \Delta \cdot \log \left(\frac{ ({1-\eta})  y_1 +  \Delta \cdot \left(\eta y_0 + ({1-\eta}) {y_1}\right)}{\eta y_0-  \Delta\cdot \left(\eta y_0 + ({1-\eta}){y_1}\right)} \right) 
\end{split}
\end{equation}

Thus, the difference in expected rewards when deviating from truthtelling is:
\begin{align}\small
    \mathbb{E}[S_{CE}({\bf q}^{(t)}) | \hat{X}^{(t)}] - \mathbb{E}[S_{CE}({\bf p}^{(t)}) | \hat{X}^{(t)}] \leq  \Delta \cdot \log \left(\frac{ ({1-\eta})  y_1 +  \Delta \cdot \left(\eta y_0 + ({1-\eta}) {y_1}\right)}{\eta y_0-  \Delta\cdot \left(\eta y_0 + ({1-\eta}){y_1}\right)} \right)  
\end{align}
\end{proof}

\paragraph{Theorem \ref{thm:exactmsr}.}
    \emph{Suppose agent $t$ is paid based on a cross-entropy market scoring rule with respect to agent $r$ who reports truthfully. Then, agent $t$ strictly maximizes their expected payoff by reporting truthfully if the reference agent has at least $k$ signals where}
    \begin{equation}\small
        k  \geq \frac{1}{-\log \left( 1-\delta \right)}\log\left(\frac{\left|\log \left(\frac{1-\eta}{\eta}\right)\right|\left({{\frac{1-\eta}{\eta}}} - {{\frac{\eta}{1-\eta}}} \right)}{8 \left(\tau\eta(1-\eta)\right)^2} \right)
    \end{equation}

\begin{proof}
We proceed similar to our proof of Theorem \ref{thm:exact}:

In order for truthful reporting to pay more than any strategic misreport, we consider the difference in expected scores when reporting the truth ${\bf p}^{(t)}$ relative to reporting some other potentially better report ${\bf q}^{(t)}$ (that is consistent with signal $\tilde{x}^{(t)}$:
\begin{equation}
\begin{split}
     \mathbb{E}[S_{CEM}({\bf q}^{(t)}) | \hat{X}^{(t)}] - \mathbb{E}[S_{CEM}({\bf p}^{(t)}) | \hat{X}^{(t)}]
     &= -H\left(\mathbb{E} \left[ \tilde{\bf p}^{(r)} \left.\right |\hat{X}^{(t)} \right], {\bf q}^{(t)} \right)  + H\left(\mathbb{E} \left[ \tilde{\bf p}^{(r)} \left.\right |\hat{X}^{(t)} \right], {\bf y}^{(t-1)} \right) \\
     &\phantom{=} + H\left({\bf p}^{(t)} \right)  - H\left({\bf p}^{(t)}, {\bf y}^{(t-1)} \right)\\
     &= \Delta \log \left(\frac{q_1^{(t)}}{1-q_1^{(t)}}\right) - \text{KL}({\bf p}^{(t)} \| {\bf q}^{(t)})  + \Delta \log\left(\frac{1-y_1^{(t-1)}}{y_1^{(t-1)}}\right)\\
     &= \Delta \log \left(\frac{q_1^{(t)}}{1-q_1^{(t)}} \cdot \frac{1-y_1^{(t-1)}}{y_1^{(t-1)}}\right) - \text{KL}({\bf p}^{(t)} \| {\bf q}^{(t)}) \\
     &= \Delta \left(\lambda_q - \lambda_{y^{(t-1)}}\right) - \text{KL}({\bf p}^{(t)} \| {\bf q}^{(t)}) \\
     &= \Delta \lambda_{\tilde{x}} - \text{KL}({\bf p}^{(t)} \| {\bf q}^{(t)})
\end{split}
\end{equation}
since $\lambda_q = \lambda_{\tilde{x}} + \lambda_{y^{(t-1)}}$ with $\lambda_{\tilde{x}} = \log\left(\frac{\mathbb{P}(\tilde{x}^{(t)}|Y=1)}{\mathbb{P}(\tilde{x}^{(t)}|Y=0)}\right)$. Thus, truthful reporting yields higher utility over reporting ${\bf q}^{(t)}$ when:
\begin{equation}
    \Delta \lambda_{\tilde{x}}  - \text{KL}({\bf p}^{(t)} \| {\bf q}^{(t)}) < 0
\end{equation}

As discussed in the proof of Theorem \ref{thm:exact}, truthful reporting yields higher utility over reporting ${\bf q}^{(t)}$ when:
\begin{equation}\small
    |\Delta| < \frac{\text{KL}({\bf p}^{(t)} \| {\bf q}^{(t)})}{\left|\lambda_{\tilde{x}}\right|}
\end{equation}

Since $\left|\lambda_{\tilde{x}}\right| \leq \left| \log\left(\frac{1-\eta}{\eta}\right) \right|$, we can write:

\begin{equation}\small
    \frac{\text{KL}({\bf p}^{(t)} \| {\bf q}^{(t)})}{|\lambda_{\tilde{x}}|} \geq \frac{2 \left(\eta(1-\eta)\tau\right)^2}{\left|\log \left(\frac{1-\eta}{\eta} \right)\right|}
\end{equation}

This gives us an upper bound for $\Delta$ that will ensure strict truthfulness:

\begin{equation}\small
    \begin{split}
        |\Delta| \leq \frac{2 \left(\eta(1-\eta)\tau\right)^2}{\left|\log \left(\frac{1-\eta}{\eta} \right)\right|}
    \end{split}
\end{equation}

This lets us identify the needed $k$:
\begin{equation}\small
    \begin{split}
         k  &\geq \frac{1}{-\log \left( 1-\delta \right)}\log\left(\frac{\left|\log \left(\frac{1-\eta}{\eta} \right)\right|\left({{\frac{1-\eta}{\eta}}} - {{\frac{\eta}{1-\eta}}} \right)}{8 \left(\eta(1-\eta)\tau\right)^2} \right)
    \end{split}
\end{equation}
\end{proof}

\paragraph{Theorem \ref{thm:pmcemsr}.}
        \emph{If agent $t$ is paid based on a cross-entropy market scoring rule with respect to agent $r$ who reports truthfully, the difference in expected payoff when deviating from truthful report ${\bf p}^{(t)}$ associated with signal $x^{(t)}$ to misreport ${\bf q}^{(t)}$ is:}
    \begin{equation}\small 
       \mathbb{E}[S_{CEM}({\bf q}^{(t)}) | x^{(t)}, x^{(1:t-1)}] - \mathbb{E}[S_{CEM}({\bf p}^{(t)}) | x^{(t)}, x^{(1:t-1)}] \leq \hat{\mathcal{D}}_{\eta}(\Delta, {\bf y}^{(t-1)})
    \end{equation}
    \emph{where ${\bf y}^{(t-1)}$ is the true market prior before agent $t$'s report and}
    \begin{equation}\small 
       \hat{\mathcal{D}}_{\eta}(\Delta, {\bf y}^{(t-1)}) = \Delta \cdot \log \left(\frac{({1-\eta}) +  \Delta\cdot\left( \eta\frac{y_0^{(t-1)}}{y_1^{(t-1)}} + ({1-\eta})\right)}{\eta -  \Delta\cdot \left(\eta + ({ 1-\eta}) \frac{y_1^{(t-1)}}{y_0^{(t-1)}}\right)} \right)
    \end{equation}

\begin{proof} Using Corollary \ref{cor:mart}, we see that when agent $t$ tells the truth, their payoff is the KL divergence of the true posterior from the previous report ${\bf y}^{(t-1)}$:

\begin{equation}\small
\begin{split}
    \mathbb{E}[S_{CEM}({\bf p}^{(t)}) | \hat{X}^{(t)}] &= -H\left(\mathbb{E} \left[ {\bf p}^{(r)} \left.\right |\hat{X}^{(t)} \right], {\bf p}^{(t)} \right) + H\left(\mathbb{E} \left[ {\bf p}^{(r)} \left.\right |\hat{X}^{(t)} \right], {\bf y}^{(t-1)} \right) \\
    &= -H\left({\bf p}^{(t)} \right) + H\left( {\bf p}^{(t)} , {\bf y}^{(t-1)}\right) \\
    &= \text{KL}\left( {\bf p}^{(t)} \,\,||\,\, {\bf y}^{(t-1)}\right) 
\end{split}
\end{equation}

Meanwhile, we can write agent $t$'s expected reward from misreporting as: 

\begin{align}\small
    \mathbb{E}[S_{CEM}({\bf q}^{(t)}) | \hat{X}^{(t)}]
    &= -H\left(\mathbb{E} \left[ \tilde{\bf p}^{(r)} \left.\right |\hat{X}^{(t)} \right], {\bf q}^{(t)} \right)  + H\left(\mathbb{E} \left[ \tilde{\bf p}^{(r)} \left.\right |\hat{X}^{(t)} \right], {\bf y}^{(t-1)} \right)\nonumber \\
    &\leq -H\left(\mathbb{E} \left[ \tilde{\bf p}^{(r)} \left.\right |\hat{X}^{(t)} \right] \right) + H\left(\mathbb{E} \left[ \tilde{\bf p}^{(r)} \left.\right |\hat{X}^{(t)} \right], {\bf y}^{(t-1)} \right) \nonumber \\
    &= \text{KL}\left( \mathbb{E} \left[ \tilde{\bf p}^{(r)} \left.\right |\hat{X}^{(t)} \right]  \,\,||\,\, {\bf y}^{(t-1)}\right)  \\
    &\leq \text{KL}\left( {\bf p}^{(t)} \,\,||\,\, {\bf y}^{(t-1)}\right) + \nabla \left( \text{KL}\left(\mathbb{E} \left[ \tilde{\bf p}^{(r)} \left.\right |\hat{X}^{(t)} \right] \right) \,\, || \,\, {\bf y}^{(t-1)}\right)^T\left( \mathbb{E} \left[ \tilde{\bf p}^{(r)} \left.\right |\hat{X}^{(t)} \right] - {\bf p}^{(t)} \right)  \tag{by convexity of KL divergence in the first argument)} \\
    &= \text{KL}\left( {\bf p}^{(t)} \,\,||\,\, {\bf y}^{(t-1)}\right) +  \Delta \cdot \log \left(\frac{({p}_1^{(t)} +  \Delta)(1-y_1^{(t-1)})}{(1 - {p}_1^{(t)} -  \Delta)y_1^{(t-1)}} \right) 
\end{align}

This expression is very similar to Equation \ref{eq:cebound}, except the log term has an additional factor depending on the previous market price as well. Under the assumption that previous agents reported truthfully, we can further simplify this expression:

\begin{equation}\small
    \begin{split}
        \Delta \cdot \log \left(\frac{({p}_1^{(t)} +  \Delta)y_0^{(t-1)}}{( {p}_0^{(t)} -  \Delta)y_1^{(t-1)}} \right)  &= \Delta \cdot \log \left(\frac{\left(\frac{\mathbb{P}(x^{(t)}|Y=1)\mathbb{P}(Y=1|x^{(s)})}{\mathbb{P}(x^{(t)} | x^{(s)})} +  \Delta\right)\mathbb{P}(Y=0|x^{(s)})}{\left(\frac{\mathbb{P}(x^{(t)}|Y=0)\mathbb{P}(Y=0|x^{(s)})}{\mathbb{P}(x^{(t)} | x^{(s)})} -  \Delta\right)\mathbb{P}(Y=1|x^{(s)})} \right) \\
        &= \Delta \cdot \log \left(\frac{\mathbb{P}(x^{(t)}|Y=1) +  \Delta\cdot \frac{\mathbb{P}(x^{(t)} | x^{(s)})}{\mathbb{P}(Y=1|x^{(s)})}}{\mathbb{P}(x^{(t)} | Y=0) -  \Delta\cdot \frac{\mathbb{P}(x^{(t)} | x^{(s)})}{\mathbb{P}(Y=0|x^{(s)})}} \right) \\
        &= \Delta \cdot \log \left(\frac{\frac{\mathbb{P}(x^{(t)}|Y=1)}{\mathbb{P}(x^{(t)} | Y=0)} +  \Delta\cdot \frac{\mathbb{P}(x^{(t)} | Y=0)\mathbb{P}(Y=0|x^{(s)}) + \mathbb{P}(x^{(t)} | Y=1)\mathbb{P}(Y=1|x^{(s)})}{\mathbb{P}(x^{(t)} | Y=0)\mathbb{P}(Y=1|x^{(s)})}}{1 -  \Delta\cdot \frac{\mathbb{P}(x^{(t)} | Y=0)\mathbb{P}(Y=0|x^{(s)}) + \mathbb{P}(x^{(t)} | Y=1)\mathbb{P}(Y=1|x^{(s)})}{\mathbb{P}(x^{(t)} | Y=0)\mathbb{P}(Y=0|x^{(s)})}} \right) \\
        &= \Delta \cdot \log \left(\frac{\frac{\mathbb{P}(x^{(t)}|Y=1)}{\mathbb{P}(x^{(t)} | Y=0)} +  \Delta\cdot\left( \frac{\mathbb{P}(Y=0|x^{(s)})}{\mathbb{P}(Y=1|x^{(s)})} + \frac{\mathbb{P}(x^{(t)} | Y=1)}{\mathbb{P}(x^{(t)} | Y=0)}\right)}{1 -  \Delta\cdot \left(1 + \frac{ \mathbb{P}(x^{(t)} | Y=1)\mathbb{P}(Y=1|x^{(s)})}{\mathbb{P}(x^{(t)} | Y=0)\mathbb{P}(Y=0|x^{(s)})}\right)} \right) \\
        &\leq \Delta \cdot \log \left(\frac{\frac{1-\eta}{\eta} +  \Delta\cdot\left( \frac{y_0^{(t-1)}}{y_1^{(t-1)}} + \frac{1-\eta}{\eta}\right)}{1 -  \Delta\cdot \left(1 + \frac{ 1-\eta}{\eta} \cdot \frac{y_1^{(t-1)}}{y_0^{(t-1)}}\right)} \right) \\
        &= \Delta \cdot \log \left(\frac{({1-\eta}) +  \Delta\cdot\left( \eta\frac{y_0^{(t-1)}}{y_1^{(t-1)}} + ({1-\eta})\right)}{\eta -  \Delta\cdot \left(\eta + ({ 1-\eta}) \frac{y_1^{(t-1)}}{y_0^{(t-1)}}\right)} \right)
    \end{split}
\end{equation}

Thus, the difference in expected rewards when deviating from truthtelling is:
\begin{equation}\small
    \mathbb{E}[S_{CEM}({\bf q}^{(t)}) | \hat{X}^{(t)}] - \mathbb{E}[S_{CEM}({\bf p}^{(t)}) | \hat{X}^{(t)}] \leq  \Delta \cdot \log \left(\frac{({1-\eta}) +  \Delta\cdot\left( \eta\frac{y_0^{(t-1)}}{y_1^{(t-1)}} + ({1-\eta})\right)}{\eta -  \Delta\cdot \left(\eta + ({ 1-\eta}) \frac{y_1^{(t-1)}}{y_0^{(t-1)}}\right)} \right)
\end{equation}

\end{proof}

\paragraph{Theorem \ref{thm:rectru}}
    \emph{In a self-resolving prediction market where a reference agent can observe at least $k$ additional predictions (with $k$ chosen according to Theorem \ref{thm:exactmsr}), the strategy profile of truthful reporting and belief profile of believing all previous agents reported truthfully is a strict Perfect Bayesian Equilibrium. When $k$ is chosen according to Remark \ref{rem:kproc}, this strategy and belief profile is a $\varepsilon-$PBE. Off the equilibrium path, agents believe any observed (mis)report to be truthful if that report is consistent with the reporting agent's signal structure. If the observed prediction is invalid, (i.e., inconsistent with any signal the reporting agent could have observed), all subsequent agents may either update as if they learned some arbitrary signal or ignore this prediction, and believe all subsequent agents continue to report truthfully.  }

\begin{proof} 

 Consider the strategy profile where all agents in the market report their beliefs truthfully and believe all other agents report their beliefs truthfully. From agent $t$'s perspective, in any supposed truthful equilibrium, they believe all previous reports to be truthful, and know that any valid report they make will be believed and any invalid report they make can induce an update consistent with an arbitrary signal in subsequent agent's beliefs. However, the reference agent can observe at least $k$ other truthful (in equilibrium) predictions that agent $t$ cannot.
 
 Suppose $k$ is chosen as in Remark \ref{rem:kproc}. Then, Theorem \ref{eq:imptheorem} guarantees that agent $t$'s expectation of their reference agent's beliefs is no more than $\varepsilon'$ away from their own true beliefs. Further, by Theorem \ref{thm:pmcemsr}, agent $t$'s maximum possible gain in deviating from truthful reporting is upper bounded, i.e., $\hat{\mathcal{D}}_\eta(\Delta, {\bf p}^{(t-1)}) \leq \varepsilon$. Thus, truthful reporting is an $\varepsilon-$PBE.

 If instead $k$ is chosen as in Theorem \ref{thm:exactmsr}, agent $t$ strictly maximizes their payoff by reporting truthfully; thus, truthful reporting is a strict PBE.
     
 \end{proof}

\end{document}